\documentclass[12pt,a4paper]{article}

\usepackage{amsmath}
\usepackage{amssymb}
\usepackage{mathrsfs}
\usepackage{amsthm}
\usepackage{graphicx}
\usepackage{color}
\usepackage{float}
\usepackage{subfigure}
\usepackage{framed}
\restylefloat{table}

\theoremstyle{definition}
	\newtheorem{definition}{Definition}[section]
	\newtheorem*{remark}{Remark}

	\newtheorem{lemma}[definition]{Lemma}
	\newtheorem{proposition}[definition]{Proposition}
	\newtheorem{theorem}[definition]{Theorem}
\theoremstyle{remark}

\title{Multi-breather solutions to the Sasa-Satsuma equation}
	
\author{Chengfa Wu$^1$,  Bo Wei$^1$,  Changyan Shi$^1$, Bao-Feng Feng$^{*2}$
\\
\\
 $^1$Institute for Advanced Study, Shenzhen University\\ Shenzhen 518060, People's Republic of China \\ $^2$School of Mathematical and Statistical Sciences, \\ The University of
Texas Rio Grande Valley Edinburg,
\\ Edinburg, TX 78541-2999, USA
}

\date{}

\begin{document}

\maketitle

\begin{figure}[b]
\rule[-2.5truemm]{5cm}{0.1truemm}\\[2mm]
{\footnotesize  *Corresponding author. Email address: baofeng.feng@utrgv.edu.}  

\end{figure}

%%%% Subject entries to be placed here %%%%
%\subject{mathematical physics, applied mathematics}

%%%% Keyword entries to be placed here %%%%

\begin{abstract}
General breather solution to the Sasa-Satsuma (SS) equation is systematically investigated in this paper.
We firstly transform the SS equation into a set of three Hirota bilinear equations under proper plane wave background.
Starting from a specially arranged tau-function of the Kadomtsev-Petviashvili hierarchy and a set of eleven bilinear equations satisfied, we implement a series steps of reduction procedure, i.e., C-type reduction, dimension reduction and complex conjugate reduction, and reduce these eleven equations to three bilinear equations for the SS equation. Meanwhile, general breather  solution to the SS equation is found in determinant of even order. The one- and two-breather solutions are calculated and analyzed in details.

\end{abstract}

\maketitle

\section{Introduction}
Breathers are ubiquitous phenomena in many physical systems either in continuous or discrete ones. It is a particular type of nonlinear wave whose energy is localized in space but oscillates over time, or vice versa. The exactly solvable sine-Gordon equation  \cite{Ablowitz} and the focusing nonlinear Schr\"odinger equation \cite{Akh1} are examples of one-dimensional partial differential equations that possess breather solutions \cite{Akh2}.

The so-called intrinsic localized modes (ILMs) or the discrete breathers (DBs) in Fermi-Pasta-Ulam (FPU) lattices were reported in the late 1980s \cite{Sievers88,Page90}.
They have been
recently observed experimentally in various physical contexts such
as coupled optical waveguides \cite{EisenbergPRL98,Sukorukov03}, Josephson
junction ladders \cite{JosephonEx,Binder00}, antiferromagnet crystals
\cite{SieversPRL99}, and  micromechanical oscillator arrays \cite{SieversPRL03}.

Breathers have met with success in understanding the final stage of a certain nonlinear process that is initiated from modulation instability (MI, also known as the Benjamin-Feir instability) \cite{Yuen1980Lake}. It is well-known that MI is one of the most ubiquitous phenomena in nature and commonly appears in many physical contexts such as water waves, plasma waves and electromagnetic transmission lines \cite{Zakharov2009Ostrovsky}. Whereas recent theoretical developments indicated that the presence of baseband MI supports the generation of rogue waves (RW) \cite{Dysthe2008KrogstadMuller}, breathers also appear to be a significant strategy in deriving RW  solutions of many integrable equations \cite{Chowdury2017KrolikowskiAkhmediev,Feng2021LingZhu}.

The nonlinear Schr\"{o}dinger equation (NLSE)
 \begin{equation}
	\begin{aligned}
    &\mathbf{i} \frac{\partial q}{\partial T}+\frac{1}{2} \frac{\partial^{2} q}{\partial X^{2}} \pm |q|^{2} q=0\\
    \end{aligned}\label{Sasa1990equation4}
\end{equation}
describes the evolution of weakly nonlinear and quasi-monochromatic waves in dispersive media \cite{Benney1967Newell}. This equation has found applications in  numerous areas of physics, ranging from nonlinear optical fibers \cite{Agrawal1995}, plasma physics \cite{Zakharov1972} to  Bose-Einstein condensates  \cite{Dalfovo1999GiorginiStringari}. From the mathematical point of view, the NLSE is considered to be a fundamental model in investigating breather and RW  solutions  \cite{Akhmediev1986Korneev,Ohta2012Yang,Chen2018Pelinovsky}. In particular,
%most attention on the study of breather solutions of the NLSE has been focused on
the Akhmediev breather (AB) \cite{Akhmediev1986Korneev} and Kuznetsov-Ma soliton (KM) \cite{Kuznetsov1977,Ma1979}, where AB (KM) is periodic in space (time) and localized in time (space), have captured wide attention. Remarkably, when we take the large-period limits, both of them degenerate to the Peregrine soliton \cite{Peregrine1983}, which is localized both in time and space and turns into a prototype of RWs. It turns out that this idea has been widely adopted in constructing RW solutions of many other integrable equations and their multi-component generalizations \cite{Feng2021LingZhu,Zhang2018Yan}.
%(localized in space and periodic in time)(localized in space and periodic in time)

The NLSE is one of the most fundamental integrable equations in the sense that it only incorporates the lowest-order dispersion and the lowest-order nonlinear term. However, higher-order terms are indispensable in more complicated circumstances, such as modeling the ultrashort pulses generated due to the MI \cite{Agrawal2011} and examing the one-dimensional Heisenberg spin chain \cite{Porsezian1992DanielLakshmanan}. As such, a number of integrable extensions of the NLSE have been proposed, including the higher-order NLSE \cite{Agrawal1995}, the Sasa-Satsuma equation \cite{Sasa1991Satsuma,Xu2013Fan} and the Kundu-NLSE \cite{Kundu1984}, to name a few examples. Therefore it is natural to expand the investigations on NLSE to these integrable models. While compared with the NLSE it is more challenging to obtain soliton, breather, or RW solutions of these equations \cite{Gedalin1997ScottBand,Mihalache1993TornerMoldoveanuPanoiuTruta,Gilson2003HietarintaNimmoOhta,Shi2019LiWu}, the occurrence of higher-order terms may also induce various new features to the solutions and enrich the solution dynamics \cite{Chowdury2017KrolikowskiAkhmediev}.

As mentioned above, the Sasa-Satsuma equation (SSE) is a nontrivial integrable extension of the NLSE and can be written in the form \cite{Sasa1991Satsuma}
\begin{equation} \label{Sasa1990equation4}
	\begin{aligned}
    &\mathbf{i} \frac{\partial q}{\partial T}+\frac{1}{2} \frac{\partial^{2} q}{\partial X^{2}}+|q|^{2} q+\mathbf{i} \varepsilon\left\{\frac{\partial^{3} q}{\partial X^{3}}+6|q|^{2} \frac{\partial q}{\partial X}+3 q \frac{\partial|q|^{2}}{\partial X}\right\}=0,\\
    \end{aligned}
\end{equation}
where $q$ corresponds to the complex envelope of the wave
field and the real constant $\varepsilon$  scales the integrable perturbations of the NLSE. For $\varepsilon=0$, the SSE reduces to the NLSE. As an extension of the NLSE,  the SSE consists of terms describing the third-order dispersion, the self-steepening and the self-frequency shift that are commonly involved in nonlinear optics \cite{Mihalache1997TrutaCrasovan,Solli2007RopersKoonathJalali}. For the convenience of analyzing the SSE, according to the work of Sasa and Satsuma \cite{Sasa1991Satsuma},  one can introduce the transformation
\begin{equation}
%\left\{
\begin{aligned}
&u(x, t)=q(X, T) \exp \left\{-\frac{\mathbf{i}}{6 \varepsilon}\left(X-\frac{T}{18 \varepsilon}\right)\right\}, \\
%&t=T, \\
%&x=X-\frac{T}{12 \varepsilon},
\end{aligned}
%\right.
\end{equation}
where $t=T$ and $x=X-T/(12 \varepsilon)$, then the equation   \eqref{Sasa1990equation4} is transformed into
\begin{equation}
u_t + \varepsilon (u_{xxx} + 6 |u|^2 u_x + 3 u (|u|^2)_x) = 0.\label{Sasa1990equation6}
\end{equation}
On account of its integrability and physical implications, the SSE has attracted much attention since it was discovered. For instance, the double hump soliton solution of the SSE was obtained by Mihalache et al. \cite{Mihalache1993TornerMoldoveanuPanoiuTruta} while its multisoliton solutions  have been constructed in the Refs.\cite{Gilson2003HietarintaNimmoOhta,Ohta2010} by the Kadomtsev-Petviashvili (KP) hierarchy reduction method. In addition to the soliton solutions, RW solutions \cite{Chen2013,Akhmedieva2015Soto-CrespoDevineHoffmannc,Mu2016Qin,Ling2016} of the SSE have also been found via the method of Darboux transformation \cite{Zhang2018Yan}, and in contrast to the NLSE, several intriguing solution structures were reported like the so-called twisted RW pair \cite{Chen2013}. Beyond that, the long-time asymptotic behaviour of the SSE with decaying initial data was analyzed in \cite{Liu2018GengXue} by formulating the Riemann-Hilbert problem.

Despite extensive investigations on the SSE, its breather solutions have not been systematically examined, to the best of our knowledge. Consequently, the main objective of this paper is to derive multi-breather solutions to the Sasa-Satsuma equation
\begin{equation}\label{SS equation}
u_{t}=u_{x x x}-6 c|u|^{2} u_{x}-3 c u\left(|u|^{2}\right)_{x},
\end{equation}
where $c$ is a real constant. The rest of this paper is organized as follows. In Section \ref{Solutions of the Sasa-Satsuma equation}, general multi-breather solutions of equation \eqref{SS equation} are presented in Theorem \ref{thm}. The detailed derivations of these solutions are provided in Section \ref{Derivation of Solutions of the Sasa-Satsuma equation}. In this process, we firstly transform the equation \eqref{SS equation} into bilinear forms. Then multi-breather solutions of equation \eqref{SS equation} can be obtained by relating the bilinear forms of \eqref{SS equation} with a set of eleven bilinear equations in the KP hierarchy. Although the idea seems to be straightforward, the intermediate computations are extremely complicated due to the complexity of the SSE and multiple corresponding bilinear equations from the KP hierarchy. In addition to the dimension reduction and the complex conjugate reduction, which are the common obstructions in applying the KP hierarchy reduction method \cite{Ohta2012Yang,Chen2018FengMarunoOhta,Rao2017PorsezianHeKanna,Feng2018LuoAblowitzMusslimani,Chen2019ChenFengMaruno,Li2020FuWu}, a new obstacle is to tackle the symmetry reduction \eqref{symmetry2}. As pointed out in \cite{Gilson2003HietarintaNimmoOhta}, when applying the direct method \cite{Hirota2004} to find soliton solutions, one only needs to truncate at power two of the formal expansion for NLSE whereas one has to go to  power four for SSE, thereby resulting in more sophisticated analysis. It turns out that this also appears in our consideration, namely the structure of breather solutions of SSE is more intricate than that of NLSE (see Theorem \ref{thm}).  In Section \ref{Dynamics of breather solutions}, the solution dynamics are discussed in detail. Six types of first-order breathers were found totally and various configurations of second- and third-order breathers have been illustrated.  The main results of this paper are summarized in Section 5.

\section{Multi-breather solutions to the Sasa-Satsuma  equation}\label{Solutions of the Sasa-Satsuma equation}
In this section, we present the multi-breather solutions to the Sasa-Satsuma equation \eqref{SS equation}.

\begin{theorem}\label{thm}
The Sasa-Satsuma  equation \eqref{SS equation}
admits the multi-breather solution
\begin{equation}\label{SS solution}
 u=\frac{g}{f} e^{\mathrm{i}\left(\kappa(x-6 c t)-\kappa^{3} t\right)}
\end{equation}
where $\kappa$ is real,
$$f(x,t)=\tau_{0} (x-6 c t,t), \quad g(x,t)=\tau_{1} (x-6 c t,t)$$ and $\tau_{k} \, (k=0,1)$ is defined as
\begin{equation}\label{SS solution_tau fucntion}
  \tau_{k}=\left| \sum_{m,n=1}^{2} \frac{1}{p_{i m}+p_{j n}}\left(-\frac{p_{i m}-a}{p_{j n}+a}\right)^{k} e^{\xi_{i m}+\xi_{j n}}\right|_{2 N \times 2 N}.
\end{equation}
Here, $a = \mathbf{i} \kappa$ is purely imaginary, $ \xi_{im}=p_{im} x +p_{im}^{3} t+ \xi_{im,0}$, $N$ is a positive integer and the parameters $  \xi_{im,0}, p_{im} \, (i=1, \cdots, 2N, m = 1,2)$  satisfy the constraints
 \begin{eqnarray} \label{parameter constraint1}
   \left(p_{i 1}^{2}+\kappa^{2}\right)\left(p_{i 2}^{2}+\kappa^{2}\right)&=&  -2 c\left(p_{i 1} p_{i 2}-\kappa^{2}\right) \end{eqnarray}
 and
 \begin{equation}\label{parameter constraint2-1}
  p_{N+l, m}=p_{l m}^{*}, \quad   \xi_{N+l,m, 0} =  \xi_{lm, 0}^*, \quad l=1, \cdots, N,
 \end{equation}
where $*$ denotes complex conjugation.
\end{theorem}

\begin{remark} \label{number of free parameters}
We note that the parameter relations \eqref{parameter constraint1} and \eqref{parameter constraint2-1} presented in Theorem \ref{thm} give rise to $6N+1$ free real parameters which include $\kappa$, the real parts and imaginary parts of $p_{i1} $ and $\xi_{im, 0}, i=1, \cdots, N, m =1,2$.
%as free parameters for \eqref{parameter constraint1} and \eqref{parameter constraint2-1} while the real parts of $p_{i1}, i=1, \cdots, 2N$ and the real parts and imaginary parts of  $\xi_{im, 0}, i=1, \cdots, N, m =1,2$ can be freely chosen  for \eqref{parameter constraint1} and \eqref{parameter constraint2-2}. {\color{red} In addition,  is free for both cases.}
\end{remark}

\begin{remark} \label{p_i2}
For $c = -4 \kappa ^2$, we can solve the  equation \eqref{parameter constraint1}  for  $p_{i2}$  that is given by
\begin{eqnarray*}
 p_{i2} =  \frac{4 \kappa ^2 p_{i1} \pm \mathbf{i} \kappa  \left(p_{i1}^2-3 \kappa ^2\right)}{\kappa ^2+p_{i1}^2}.
\end{eqnarray*}
When $c \not= -4 \kappa ^2$, the expression of $p_{i2}$ is more complicated.
If we set $p_{i1} = P_{\text{R}}+ \mathbf{i} P_{\text{I}}$, where $P_{R}$ and  $P_{I}$ represent the real and imaginary parts of $p_{i1} $ respectively, then
%by using the De Moivre's formula,
$p_{i2}$ can be expressed as
\begin{eqnarray*}
p_{i2}  &=& \frac{G+\mathbf{i} H}{K }
%\left\{ \right.
%  \\
%  && \pm 2 \beta  P_{I} P_R  ,
\end{eqnarray*}
where
\begin{eqnarray*}
  G &=& - P_{I}^2 \left(\pm \alpha +2 c P_R\right)+\left(\kappa ^2+P_R^2\right) \left(\pm \alpha -2 c P_R\right) \\
  H &=& \left(\kappa ^2-P_{I}^2\right) \left(\pm \beta -2 c P_{I}\right)+P_R^2 \left(\pm \beta +2 c P_{I}\right) \mp 2 \alpha  P_{I} P_R
  \\
  K &=& 2 \left[2 P_{I}^2 \left(P_R^2-\kappa ^2\right)+P_{I}^4+\left(\kappa ^2+P_R^2\right){}^2\right]
\end{eqnarray*}
with
\begin{eqnarray*}
  \alpha &=& \sqrt[4]{X^2+Y^2} \cos \theta \\
  \beta &=& \sqrt[4]{X^2+Y^2} \sin \theta\\
 X &=& -4 \left[P_{I}^2 \left(c^2+2 c \kappa ^2+\kappa ^2 P_{I}^2-2 \kappa ^4\right)-P_R^2 \left(c^2+2 c \kappa ^2+6 \kappa ^2 P_{I}^2-2 \kappa ^4\right)-2 c \kappa ^4+\kappa ^6+\kappa ^2 P_R^4\right] \\
  Y &=& 8 P_{I} P_R \left[c^2+2 c \kappa ^2+2 \kappa ^2 \left(P_{I}^2-P_R^2\right)-2 \kappa ^4\right]
  \\
  \theta &=&   \arctan (Y/X)/2.
\end{eqnarray*}

\end{remark}

\section{Derivation of the multi-breather solutions}\label{Derivation of Solutions of the Sasa-Satsuma equation}

This section is devoted to the construction of multi-breather solutions to the Sasa-Satsuma equation \eqref{SS equation}. It consists of two main steps. First, we transform the Sasa-Satsuma equation \eqref{SS equation} into bilinear forms. Then multi-breather solutions are derived by showing that such bilinear equations can be obtained from reductions of the KP hierarchy.
\subsection{Bilinear forms of the Sasa-Satsuma equation}
The bilinearization of the Sasa-Satsuma equation \eqref{SS equation} is established by the proposition below.

\begin{proposition} The Sasa-Satsuma equation
\begin{equation*}
u_{t}=u_{x x x}-6 c|u|^{2} u_{x}-3 c u\left(|u|^{2}\right)_{x}
\end{equation*}
can be transformed into the system of bilinear equations
\begin{equation} \label{bilinear form-SS}
\begin{aligned}
    &\left(D_{x}^{2}-4c\right) f \cdot f=-4 c g g^{*}\\
    &\left(D_{x}^{3}-D_{t}+3 \mathrm{i} \kappa D_{x}^{2}-3\left(\kappa^{2}+ 4 c \right) D_{x} - 6 \mathrm{i} \kappa c \right) g \cdot f+6 \mathrm{i} \kappa c q g=0\\
    &\left(D_{x}+2 \mathrm{i} \kappa\right) g \cdot g^{*}=2 \mathrm{i} \kappa q f\\
    \end{aligned}
\end{equation}
by the variable transformation
\begin{equation} \label{transformation1}
    u=\frac{g}{f} e^{\mathrm{i}\left(\kappa(x-6 c t)-\kappa^{3} t\right)},
\end{equation}
where $\kappa$ is real, $f$ is a real-valued function, $g$ is a complex-valued function, $q$ is an  auxiliary  function  and $D$ is the Hirota's bilinear operator \cite{Hirota2004} defined by
\begin{eqnarray*}\label{doperator}
D_x^mD_t^nf\cdot g=\left.\left(\frac{\partial}{\partial x}-\frac{\partial}{\partial {x'}}\right)^m\left(\frac{\partial}{\partial t}-\frac{\partial}{\partial{t'}}\right)^n
[f(x,t)g(x',t')]\right|_{x'=x,t'=t}.
\end{eqnarray*}

\end{proposition}

\begin{proof}   By substituting \eqref{transformation1} into equation \eqref{SS equation} and rewriting the resulting equation in bilinear forms, we obtain

\begin{equation}
   \begin{aligned}
&-f^{2} D_{t} g \cdot f + f^{2} D_{x}^{3} g \cdot f- 3\left(D_{x} g \cdot f\right)\left(D_{x}^{2} f \cdot f\right) + 3 \mathbf{i} \kappa f^{2}\left(D_{x}^{2} g \cdot f\right)\\
&-3 \mathbf{i} \kappa f g\left(D_{x}^{2} f \cdot f\right)-3 \kappa^{2} f^{2}\left(D_{x} g \cdot f\right) -c
\left(9 g g^{*} D_{x} g \cdot f+3 g^{2} D_x g^{*} \cdot f\right)\\
&-c\left(-6 \mathbf{i} \kappa g f^{3}+6 \mathbf{i} \kappa f g^{2} g^{*}\right)=0.\label{eq:8}
    \end{aligned}
\end{equation}
Apply the following identity to the equation above
\begin{equation}
   \begin{aligned}
   9 g g^{*} D_{x} g \cdot f+3 g^{2} D_x g^{*} \cdot f
   =-3 g f(D_x g \cdot g^*) + 12 g g^*(D_x g \cdot f),
    \end{aligned}
\end{equation}
then (\ref{eq:8}) can be rearranged as
%\begin{eqnarray}
%      0&=&f^2[(D_x^3-D_t+3 i \alpha D_x^2 -3 \alpha^2 D_x- 6 i \alpha)g \cdot f]\\
%      &&-3 g f[(D_x-2 i \alpha)g \cdot g^*]\\
%     &&-3(D_x g \cdot f)[(D_x^2)f \cdot f -4 g g^*]\\
%     &&-3 i \alpha f g (D_x^2 f \cdot f)
%\end{eqnarray}
   \begin{eqnarray}
   %\implies
    f^2[(D_x^3-D_t+3 \mathbf{i} \kappa D_x^2 -3( \kappa^2+4c) D_x +6\mathbf{i}c   \kappa)g \cdot f]  +3c  g f[(D_x-2 \mathbf{i} \kappa)g \cdot g^*] \nonumber \\
    -3(D_x g \cdot f)[ (D_x^2  -4c )f \cdot f  +4c  g g^*]-3 \mathbf{i} \kappa f g (D_x^2 f \cdot f)=0. \label{biliear-2}
   \end{eqnarray}
%\end{equation}
If we require
\begin{equation} \label{bilinear eq1}
  % \begin{aligned}
  (D_x^2  -4c )f \cdot f +4c  g g^*=0,
    %\end{aligned}
\end{equation}
then equation \eqref{biliear-2} reduces to

\begin{equation}
\begin{aligned}
%\implies
f^2[(D_x^3-D_t+3 \mathbf{i} \kappa D_x^2 -3 (\kappa^2+4c) D_x - 6 \mathbf{i} c \kappa)g \cdot f]
+3c g f[(D_x +2 \mathbf{i} \kappa)g \cdot g^*] = 0,
\end{aligned}
\end{equation}
which can be decomposed as
\begin{equation} \label{bilinear eq2}
\begin{cases}
(D_x^3-D_t+3 \mathbf{i} \kappa D_x^2 -3 (\kappa^2+4c) D_x- 6 \mathbf{i} c \kappa)g \cdot f= -6 \mathrm{i} \kappa c q g\\
(D_x+2 \mathbf{i} \kappa)g \cdot g^*=2 \mathbf{i} \kappa q f
    \end{cases}
\end{equation}
where $q$ is an  auxiliary  function. As a consequence, combing the equations \eqref{bilinear eq1} and \eqref{bilinear eq2} shows that the Sasa-Satsuma equation \eqref{SS equation} can be transformed into the system of bilinear equations \eqref{bilinear form-SS} via the transformation \eqref{transformation1}.
\end{proof}

\subsection{Derivation of multi-breather solutions}
In order to derive multi-breather solutions of the Sasa-Satsuma equation \eqref{SS equation}, we first present a crucial lemma.
\begin{lemma}
The bilinear equations in the KP hierarchy
\begin{eqnarray}
&&(D_{r} D_{x}-2) \tau_{k l} \cdot \tau_{k l}=-2 \tau_{k+1, l} \tau_{k-1, l} \label{KP-1} \\
&&(D_{s} D_{x}-2) \tau_{k l} \cdot \tau_{k l}=-2 \tau_{k, l+1} \tau_{k, l-1} \label{KP-2} \\
&&(D_{x}^{2}-D_{y}+2 a D_{x}) \tau_{k+1, l} \cdot \tau_{k l}=0 \label{KP-3}\\
&&(D_{x}^{2}-D_{y}+2 b D_{x}) \tau_{k, l+1} \cdot \tau_{k l}=0 \label{KP-4}\\
&&\left(D_{x}^{3}+3 D_{x} D_{y}-4 D_{t}+3 a\left(D_{x}^{2}+D_{y}\right)+6 a^{2} D_{x}\right) \tau_{k+1, l} \cdot \tau_{k l}=0 \label{KP-5}\\
&&\left(D_{x}^{3}+3 D_{x} D_{y}-4 D_{t}+3 b\left(D_{x}^{2}+D_{y}\right)+6 b^{2} D_{x}\right) \tau_{k, l+1} \cdot \tau_{k l}=0 \label{KP-6}\\
&&\left(D_{r}\left(D_{x}^{2}-D_{y}+2 a D_{x}\right)-4 D_{x}\right) \tau_{k+1, l} \cdot \tau_{k l}=0 \label{KP-7}\\
&&\left(D_{s}\left(D_{x}^{2}-D_{y}+2 b D_{x}\right)-4 D_{x}\right) \tau_{k, l+1} \cdot \tau_{k l}=0 \label{KP-8}\\
&&\left(D_{s}\left(D_{x}^{2}-D_{y}+2 a D_{x}\right)-4\left(D_{x}+a-b\right)\right) \tau_{k+1, l} \cdot \tau_{k l}+4(a-b) \tau_{k+1, l+1} \tau_{k, l-1}=0 \color{white}{ aaaaa} \label{KP-9}\\
&&\left(D_{r}\left(D_{x}^{2}-D_{y}+2 b D_{x}\right)-4\left(D_{x}+b-a\right)\right) \tau_{k, l+1} \cdot \tau_{k l}+4(b-a) \tau_{k+1, l+1} \tau_{k-1, l}=0 \color{white}{ aaaaa} \label{KP-10}\\
&&\left(D_{x}+a-b\right) \tau_{k+1, l} \cdot \tau_{k, l+1}=(a-b) \tau_{k+1, l+1} \tau_{k l} \label{KP-11}
\end{eqnarray}
admit the $M\times M$ Gram-type determinant solutions
\begin{equation} %\label{}
  \tau_{k l}=\left|m_{i j}^{k l}\right|_{M\times M}
\end{equation}
where
\begin{eqnarray}
 m_{i j}^{k l} &=& \int\left(\sum_{m=1}^{2} \phi_{i m}^{k, l}\right)\left(\sum_{n=1}^{2} \bar{\phi}_{j n}^{k, l}\right) d x \\
 &=&   \sum_{m,n=1}^{2}   \frac{1}{p_{im}+q_{jn}}\left(\frac{a-p_{im}}{a+q_{jn}}\right)^{k}\left(\frac{b-p_{im}}{b+q_{jn}}\right)^{l} e^{\xi_{im}+\bar{\xi}_{j n}}
 \\
  \phi_{i m}^{k, l} &=& \left(p_{i m}-a\right)^{k}\left(p_{i m}-b\right)^{l} e^{\xi_{i m}} \\
  \bar{\phi}_{j n}^{k, l} &=& (-1)^{k}\left(q_{j n}+a\right)^{-k}(-1)^{l}\left(q_{j n}+b\right)^{-l} e^{\bar{\xi}_{j n}}
\end{eqnarray}
with
\begin{eqnarray}
  \xi_{im} &=& p_{i m} x+p_{i m}^{2} y+p_{i m}^{3} t+\frac{1}{p_{i m}-a} r+\frac{1}{p_{i m}-b} s+\xi_{im, 0} \\
  \bar{\xi}_{jn} &=& q_{j n} x-q_{j n}^{2} y+q_{j n}^{3} t+\frac{1}{q_{j n}+a} r+\frac{1}{q_{j n}+b} s+\eta_{jn, 0}.
\end{eqnarray}
Here, $p_{i m},q_{j n},  \xi_{im, 0}, \eta_{jn, 0} \, (i,j=1,\cdots M, m , n = 1,2), a$ and $b$ are complex constants  while $k$ and $l$ are integers.
\end{lemma}

In what follows, we will establish the reductions from the bilinear equations \eqref{KP-1}-\eqref{KP-11} in the KP hierarchy to the bilinear equations \eqref{bilinear form-SS}, which consist of several steps. Once this is accomplished, multi-breather solutions of the Sasa-Satsuma equation \eqref{SS equation} will be derived. We start with the reduction from AKP to CKP \cite{Jimbo1983Miwa}. To this end, we take
%$\mathrm{AKP} \rightarrow \mathrm{CKP}$
$$
q_{j 1}=p_{j 1}, \quad q_{j 2}=p_{j 2},   \quad b=-a, \quad \xi_{jn, 0} = \eta_{jn, 0} %\quad {\color{red} c_{ij} = c_{ji}},
$$
where $j=1,\cdots M $ and $ n = 1,2$, then we obtain
\begin{equation*}
   \xi_{jn}(x, y, t, r, s)= \bar{\xi}_{jn} (x,-y, t, s, r).
\end{equation*}
Therefore, we have
\begin{eqnarray*}
  m_{j i}^{-l,-k}(x,-y, t, s, r) &=&   \sum_{m,n=1}^{2}   \frac{1}{p_{jm}+q_{in}}\left(\frac{a-p_{jm}}{a+q_{in}}\right)^{-l}\left(\frac{b-p_{jm}}{b+q_{in}}\right)^{-k} e^{(\xi_{jm}+\bar{\xi}_{i n}) (x,-y, t, s, r)}  \\
    &=&    \sum_{m,n=1}^{2}   \frac{1}{p_{jm}+p_{in}}\left(\frac{a-p_{jm}}{a+p_{in}}\right)^{-l}\left(\frac{a+p_{jm}}{a-p_{in}}\right)^{-k} e^{\xi_{i n}+\bar{\xi}_{j m} }
      \\
    &=&     \sum_{m,n=1}^{2}   \frac{1}{p_{jm}+p_{in}} \left(\frac{a-p_{in}}{a+p_{jm}}\right)^{k} \left(\frac{a+p_{in}}{a-p_{jm}}\right)^{l} e^{\xi_{i n}+\bar{\xi}_{j m} }
      \\
    &=&     \sum_{m,n=1}^{2}   \frac{1}{p_{im}+p_{jn}} \left(\frac{a-p_{im}}{a+p_{jn}}\right)^{k} \left(\frac{a+p_{im}}{a-p_{jn}}\right)^{l} e^{\xi_{i m}+\bar{\xi}_{j n} }
    \\
    &=& m_{i j}^{k l}(x, y, t, r, s)
\end{eqnarray*}
and
\begin{equation} \label{symmetry}
\tau_{k l}(x, y, t, r, s)=\tau_{-l, -k}(x,-y, t, s, r).
\end{equation}
%Set  $y=0$, then we have
%$$\tau_{k l}(x, 0, t, r, s)=\tau_{-l,-k}(x,0, t, s, r).$$

Next, we perform the dimension reduction. First, we rewrite $\tau_{k l}$ as
$$
\tau_{k l}=\prod_{i=1}^{M} e^{\xi_{i 2}+\bar{\xi}_{i 2}}\widetilde{\tau}_{k l}
$$
where
\begin{equation*}
\widetilde{\tau}_{k l}  = \left|\widetilde{m}_{i j}^{k l}\right|
\end{equation*}
and
\begin{eqnarray}
  \widetilde{m}_{i j}^{k l} &=&   F_{kl}(p_{i 1}, p_{j 1}) e^{\xi_{i 1}-\xi_{i 2}+\bar{\xi}_{j 1}-\bar{\xi}_{j 2}}+F_{kl}(p_{i 1}, p_{j 2}) e^{\xi_{i 1}-\xi_{i 2}} \\
    && +F_{kl}(p_{i 2}, p_{j 1}) e^{\bar{\xi}_{j 1}-\bar{\xi}_{j 2}}+F_{kl}(p_{i 2}, p_{j 2})
\end{eqnarray}
with
\begin{eqnarray*}
  F_{kl}(p,q) &=& \frac{1}{p+q}\left(\frac{a-p}{a+q}\right)^{k}\left(\frac{a+p}{a-q}\right)^{l} \\
  \xi_{i 1}-\xi_{i 2} &=& \left(p_{i 1}-p_{i 2}\right) x + \left(p_{i 1}^2-p_{i 2}^2\right) y +\left(p_{i 1}^3-p_{i 2}^3\right) t+\left(\frac{1}{p_{i 1}-a}-\frac{1}{p_{i 2}-a}\right) r
  \\&&+\left(\frac{1}{p_{i 1}+a}-\frac{1}{p_{i 2}+a}\right) s + \xi_{i1, 0}-\xi_{i2, 0} \\
  \bar{\xi}_{i 1}-\bar{\xi}_{i 2} &=& \left(p_{i 1}-p_{i 2}\right) x - \left(p_{i 1}^2-p_{i 2}^2\right) y +\left(p_{i 1}^3-p_{i 2}^3\right) t+   \left(\frac{1}{p_{i 1}+a}-\frac{1}{p_{i 2}+a}\right) r
  \\
  && +\left(\frac{1}{p_{i 1}-a}-\frac{1}{p_{i 2}-a}\right) s + \xi_{i1, 0}-\xi_{i2, 0}.
\end{eqnarray*}
Note that
\begin{eqnarray*}
   \left(\partial_{r}+\partial_{s}-\frac{1}{c} \partial_x \right)\widetilde{m}_{i j}^{k l} &=& [G(p_{i1},p_{i2}) + G(p_{j1},p_{j2}) ]   F(p_{i 1}, q_{j 1}) e^{\xi_{i 1}-\xi_{i 2}+\bar{\xi}_{j 1}-\bar{\xi}_{j 2}}
   \\
    &&+G(p_{i1},p_{i2})F(p_{i 1}, q_{j 2}) e^{\xi_{i 1}-\xi_{i 2}}  +G(p_{j1},p_{j2})F(p_{i 2}, q_{j 1}) e^{\bar{\xi}_{j 1}-\bar{\xi}_{j 2}},
\end{eqnarray*}
where %$c$ is a constant and
\begin{eqnarray*}
  G(p,q) &=& \frac{1}{p-a} + \frac{1}{p+a} - \frac{1}{q-a} - \frac{1}{q+a} - \frac{1}{c} \left(p-q\right)
  \\
  &=& (q-p) \left[\frac{1}{\left(p-a\right)\left(q-a\right)}+\frac{1}{\left(p+a\right)\left(q+a\right)}+\frac{1}{c}  \right].
\end{eqnarray*}
Therefore, by taking
$$
\frac{1}{\left(p_{i 1}-a\right)\left(p_{i 2}-a\right)}+\frac{1}{\left(p_{i 1}+a\right)\left(p_{i 2}+a\right)}+\frac{1}{c}=0,
$$
which is equivalent to
$$
\left(p_{i 1}^{2}-a^{2}\right)\left(p_{i 2}^{2}-a^{2}\right)+2 c\left(p_{i 1} p_{i 2}+a^{2}\right) = 0,
$$
we have
\begin{eqnarray}
 \left(\partial_{r}+\partial_{s}\right) \widetilde{\tau}_{k l}&=& \sum_{i,j=1}^M \Delta_{ij}  \nonumber \left(\partial_{r}+\partial_{s}\right)\widetilde{m}_{i j}^{k l}
 \\
    &=&   \frac{1}{c}\sum_{i,j=1}^M \Delta_{ij}   \partial_{x} \widetilde{m}_{i j}^{k l}  \nonumber
     \\
    &=& \frac{1}{c}  \partial_{x} \widetilde{\tau}_{k l}, \label{dimension reduction}
\end{eqnarray}
where $\Delta_{ij}$ denotes the $(i,j)$-cofactor of the matrix $(\widetilde{m}_{i j}^{k l})$.
Thus, with   \eqref{dimension reduction}, we can replace the derivatives in $r$ and $s$ by derivatives in $x$ in the bilinear equations \eqref{KP-1}-\eqref{KP-11}  and obtain
\begin{eqnarray}
&& \left(D_{x}^{2}-4 c\right) \widetilde{\tau}_{k l} \cdot \widetilde{\tau}_{k l}=-2 c\left(\widetilde{\tau}_{k+1, l} \widetilde{\tau}_{k-1, l}+\widetilde{\tau}_{k, l+1} \widetilde{\tau}_{k, l-1}\right) \label{reduced bilinear eq1}\\
&& \left(D_{x}^{3}-D_{t}+3 a D_{x}^{2}+3\left(a^{2}-2 c\right) D_{x}-6 a c\right) \widetilde{\tau}_{k+1, l} \cdot \widetilde{\tau}_{k l}+6 a c \widetilde{\tau}_{k+1, l+1} \widetilde{\tau}_{k, l-1}=0 \label{reduced bilinear eq2}\\
&& \left(D_{x}^{3}-D_{t}-3 a D_{x}^{2}+3\left(a^{2}-2 c\right) D_{x}+6 a c\right) \widetilde{\tau}_{k, l+1} \cdot \widetilde{\tau}_{k l}-6 a c \widetilde{\tau}_{k+1, l+1} \widetilde{\tau}_{k-1, l}=0 \label{reduced bilinear eq3}\\
&& \left(D_{x}+2 a\right) \widetilde{\tau}_{k+1, l} \cdot \widetilde{\tau}_{k, l+1}=2 a \widetilde{\tau}_{k+1, l+1} \widetilde{\tau}_{k l}. \label{reduced bilinear eq4}
\end{eqnarray}
Among the above bilinear equations, the equation \eqref{reduced bilinear eq1} is derived from bilinear equations \eqref{KP-1}-\eqref{KP-2} and  \eqref{dimension reduction} while the bilinear equation \eqref{reduced bilinear eq4} is obtained from the bilinear equation \eqref{KP-11} with $b=-a$. In view of \eqref{dimension reduction} and $b=-a$, the bilinear equations \eqref{reduced bilinear eq2} and \eqref{reduced bilinear eq3} can be derived respectively as follows
\begin{eqnarray*}
&&\dfrac{1}{c}\left[3a \times \eqref{KP-3}+\eqref{KP-5}\right]+3\times(\eqref{KP-7}+\eqref{KP-9}) = 4\times \eqref{reduced bilinear eq2}\\
&& \dfrac{1}{c}\left[3a \times \eqref{KP-4}+\eqref{KP-6}\right]+3\times(\eqref{KP-8}+\eqref{KP-10}) = 4\times \eqref{reduced bilinear eq3}.
\end{eqnarray*}
Since the bilinear equations \eqref{reduced bilinear eq1}-\eqref{reduced bilinear eq4} do not involve derivatives with respect to $y,r$ and $s$, we may take $y=r=s=0$. Then according to \eqref{symmetry}, we have
\begin{equation} \label{symmetry2}
\widetilde{\tau}_{k l}(x, t)=\widetilde{\tau}_{-l, -k}(x, t).
\end{equation}

Finally, we consider the complex conjugate reduction. Let the size of the matrix $(\widetilde{m}_{i j}^{k l})$ be even, i.e., $M = 2N$ and $a=\mathrm{i} \kappa$  be purely imaginary. Further, by imposing the parameter relations
\begin{equation}%\label{}
  p_{N+j, 1}=p_{j 1}^{*}, \quad p_{N+j, 2}=p_{j 2}^{*},  \quad \xi_{N+j,1, 0} = \xi_{j1, 0}^*,  \quad \xi_{N+j,2, 0} = \xi_{j2, 0}^*, \quad j=1, \cdots, N,
\end{equation}
  we obtain
\begin{eqnarray*}
  \xi_{j n}^*= \xi_{N+j, n}, \quad \bar{\xi}_{j n}^*= \bar{\xi}_{N+j, n}, \quad n = 1,2,
   % &=&   \\
%    &=&
\end{eqnarray*}
and
\begin{equation*}
 F_{0k}^*(p,q) = F_{k0}(p^*,q^*).
\end{equation*}
Then it yields that
\begin{eqnarray*}
(   \widetilde{m}_{i j}^{0k})^* &=&  F_{0k}^*(p_{i 1}, p_{j 1}) e^{\xi_{i 1}^*-\xi_{i 2}^*+\bar{\xi}_{j 1}^*-\bar{\xi}_{j 2}^*}+F_{0k}^*(p_{i 1}, p_{j 2}) e^{\xi_{i 1}^*-\xi_{i 2}^*} \\
    && +F_{0k}^*(p_{i 2}, p_{j 1}) e^{\bar{\xi}_{j 1}^*-\bar{\xi}_{j 2}^*}+F_{0k}^*(p_{i 2}, p_{j 2})
    \\
    &=&  F_{k0}(p_{N+i, 1}, p_{N+j, 1}) e^{\xi_{N+i, 1}-\xi_{N+i, 2}+\bar{\xi}_{N+j, 1}-\bar{\xi}_{N+j, 2}}+F_{k0}(p_{N+i, 1}, p_{N+j, 2}) e^{\xi_{N+i, 1}-\xi_{N+i, 2}} \\
    && +F_{k0}(p_{N+i, 2}, p_{N+j, 1}) e^{\bar{\xi}_{N+j, 1}-\bar{\xi}_{N+j, 2}}+F_{k0}(p_{N+i, 2}, p_{N+j, 2})
    \\
     &=&  \widetilde{m}_{N+i,N+ j}^{k 0}.
\end{eqnarray*}
With similar argument, we can obtain
\begin{eqnarray*}
%(   \widetilde{m}_{i j}^{0k})^* &=&  \widetilde{m}_{N+i,N+ j}^{k 0} \\
(   \widetilde{m}_{i, N+j}^{0k})^* &=&  \widetilde{m}_{N+i, j}^{k 0} \\
 (   \widetilde{m}_{N+i, j}^{0k})^* &=&  \widetilde{m}_{i,N+ j}^{k 0} \\
(   \widetilde{m}_{N+i, N+j}^{0k})^* &=&  \widetilde{m}_{i, j}^{k 0}
\end{eqnarray*}
and hence
\begin{eqnarray*}
  \widetilde{\tau}^*_{0 k } &=&  \left|\begin{array}{cc}
(   \widetilde{m}_{i j}^{0k})^* & (   \widetilde{m}_{i, N+j}^{0k})^* \\
(   \widetilde{m}_{N+i, j}^{0k})^*  & (   \widetilde{m}_{N+i, N+j}^{0k})^*
\end{array}\right| \\
&=&  \left|\begin{array}{cc}
\widetilde{m}_{N+i,N+ j}^{k 0} & \widetilde{m}_{N+i, j}^{k 0} \\
\widetilde{m}_{i,N+ j}^{k 0}  & \widetilde{m}_{i, j}^{k 0}
\end{array}\right| \\
&=&  \left|\begin{array}{cc}
\widetilde{m}_{i, j}^{k 0} & \widetilde{m}_{i,N+ j}^{k 0} \\
  \widetilde{m}_{N+i, j}^{k 0}  &  \widetilde{m}_{N+i,N+ j}^{k 0}
\end{array}\right| \\
\\
&=& \widetilde{\tau}_{ k 0} .
\end{eqnarray*}
On the other hand, using similar method as above, we can prove that
$$
\widetilde{\tau}^*_{k k } =  \widetilde{\tau}_{ k k},
$$
which implies that $\widetilde{\tau}_{ k k}$ is real.
Define
$$
\widetilde{f}=\widetilde{\tau}_{00}, \quad \widetilde{g}=\widetilde{\tau}_{10}, \quad \widetilde{h}=\widetilde{\tau}_{01}, \quad \widetilde{q}=\widetilde{\tau}_{11},
$$
then we find that  $\widetilde{f}$ and $\widetilde{q}$ are real-valued functions and $\widetilde{g}^*=\widetilde{h}$. According to \eqref{symmetry2}, we have
\begin{equation*}
  \widetilde{\tau}_{-1,0} = \widetilde{g}^*, \quad \widetilde{\tau}_{0,-1} = \widetilde{g}.
\end{equation*}
Therefore, the bilinear equations \eqref{reduced bilinear eq1}-\eqref{reduced bilinear eq4} become
\begin{equation} \label{bilinear-2c}
   \left\{\begin{array}{l}
\left(D_{x}^{2}-4 c\right) \widetilde{f} \cdot \widetilde{f}=-4 c \widetilde{g}\, \widetilde{g}^{*} \\
\left(D_{x}^{3}-D_t+3 i \kappa D_{x}^{2}-3\left(\kappa^{2}+2 c\right) D_{x}-6 i \kappa c\right) \widetilde{g} \cdot \widetilde{f}+6 i \kappa c \widetilde{q} \, \widetilde{g}=0 \\
\left(D_{x}+2 i \kappa\right) \widetilde{g} \cdot \widetilde{g}^{*}=2 \mathrm{i} \kappa \widetilde{q} \widetilde{f}.
\end{array}\right.
\end{equation}
Let
\begin{eqnarray*}
 \widehat{f}(x,t) = \widetilde{f}(x-6ct,t), \quad  \widehat{g}(x,t) = \widetilde{g}(x-6ct,t), \quad  \widehat{q}(x,t) = \widetilde{q}(x-6ct,t),
\end{eqnarray*}
then the system of bilinear equations \eqref{bilinear-2c} reduces to

%As this system of bilinear equations  coincides with
\eqref{bilinear form-SS}, and thus we can obtain the following solution   to the Sasa-Satsuma equation \eqref{SS equation}
\begin{equation} % \label{transformation1}
    u=\frac{\widehat{g}}{\widehat{f}} e^{\mathrm{i}\left(\kappa(x-6 c t)-\kappa^{3} t\right)},
\end{equation}
where
\begin{equation*}
 \widehat{f}(x,t) = \widetilde{\tau}_{00}(x-6ct,t), \quad  \widehat{g} = \widetilde{\tau}_{10}(x-6ct,t).
\end{equation*}
In addition, let
 \begin{eqnarray*}
   f(x,t) &=& \prod_{i=1}^{2N} e^{(\xi_{i 2}+\bar{\xi}_{i 2})(x-6ct, t)} \widehat{f}(x,t) = \tau_{00}(x-6ct,t)  \\
   g(x,t) &=& \prod_{i=1}^{2N} e^{(\xi_{i 2}+\bar{\xi}_{i 2})(x-6ct, t)} \widehat{g}(x,t) = \tau_{10}(x-6ct,t)
 \end{eqnarray*}
then it is found that
\begin{equation} % \label{transformation1}
    u=\frac{g}{f} e^{\mathrm{i}\left(\kappa(x-6 c t)-\kappa^{3} t\right)},
\end{equation}
where
\begin{eqnarray*}
  f(x,t) =  \tau_{0}(x-6ct,t), \quad g(x,t) =  \tau_{1}(x-6ct,t)
\end{eqnarray*}
and
\begin{equation*}
   \tau_{k} =  \tau_{k0}, \quad k = 0,1
\end{equation*}
also solves the Sasa-Satsuma equation \eqref{SS equation}. Thus the proof is completed.

\section{Dynamics of breather solutions} \label{Dynamics of breather solutions}

In this section, we discuss the dynamics of the breather solutions of the Sasa-Sastuma equation derived in Theorem \ref{thm}.

\subsection{First-order breather solutions} \label{First-order breather solution}

To obtain the first-order breather solutions to    equation \eqref{SS equation}, we set $N = 1$ in Theorem \ref{thm}.
In this case, we have
\begin{eqnarray*}
%\begin{aligned}
\tau_0 = \left|
\begin{array}{cc}
   m_{11}^{(0)} &m_{12}^{(0)}   \\
 m_{21}^{(0)} &m_{22}^{(0)}
\end{array}
\right|
,
\quad
\tau_1 = \left|
\begin{array}{cc}
   m_{11}^{(1)} &m_{12}^{(1)}   \\
 m_{21}^{(1)} &m_{22}^{(1)}
\end{array}
\right|,
%\end{aligned}
\end{eqnarray*}
where
\begin{equation*}%\label{}
  m_{ij}^{(k)} = \sum\limits_{m,n=1}^{2} \frac{1}{p_{im} + p_{jn}} \left(\dfrac{a-p_{1 m}}{a+p_{1 n}}\right)^k e^{\xi_{im} + \xi_{jn}}, \quad i,j=1, 2, k =0,1,
\end{equation*}
$a = \mathbf{i} \kappa$ is purely imaginary, $ \xi_{im}=p_{im} x +p_{im}^{3} t+ \xi_{im,0} \, (m=1, 2)$,
and the complex parameters $  \xi_{im,0}, p_{im} \, $  satisfy the constraints
 \begin{eqnarray*} %\label{parameter constraint1}
   \left(p_{i 1}^{2}+\kappa^{2}\right)\left(p_{i 2}^{2}+\kappa^{2}\right)=  -2 c\left(p_{i 1} p_{i 2}-\kappa^{2}\right)
     \end{eqnarray*}
and
\begin{equation*}%\label{}
    p_{2 m}=p_{1 m}^{*}, \quad   \xi_{2m, 0} =   \xi_{1m, 0}^*.
\end{equation*}

After some tedious algebra, we can express the solutions  \eqref{SS solution} in terms of trigonometric functions and hyperbolic functions
\begin{equation} \label{first-order breather}
 u=\frac{g(x, t)}{f(x, t)} e^{\mathrm{i}\left(\kappa(x-6 c t)-\kappa^{3} t\right)},
\end{equation}
with
\begin{eqnarray*}
%\begin{aligned}
f(x, t) &=& \alpha_1  + M_1\cosh (2W_1 - \theta_1) +\alpha_2  \cos (V_1) \cosh (W_1) +   \alpha_3  \cos (V_1) \sinh (W_1)  \nonumber\\
    && +  \alpha_4 \sin (V_1) \cosh (W_1) +   \alpha_5 \sin (V_1) \sinh (W_1)  + \alpha_6 \cos (2 V_1- \theta_2)
    \\
g(x, t) &=& \beta_1  + M_2\cosh (2W_1 - \theta_3) +\beta_2  \cos (V_1) \cosh (W_1) +   \beta_3  \cos (V_1) \sinh (W_1)  \nonumber\\
    && +  \beta_4 \sin (V_1) \cosh (W_1) +   \beta_5 \sin (V_1) \sinh (W_1)  + \beta_6 \cos (2 V_1- \theta_4) \\
    && + \mathbf{i} \left[ \gamma_1  + M_3\cosh (2W_1 - \theta_5) +\gamma_2  \cos (V_1) \cosh (W_1) +   \gamma_3  \cos (V_1) \sinh (W_1) \right. \nonumber\\
    && \left. +  \gamma_4 \sin (V_1) \cosh (W_1) +   \gamma_5 \sin (V_1) \sinh (W_1)  + \gamma_6 \cos (2 V_1- \theta_6) \right],
%\end{aligned}
\end{eqnarray*}
where $V_1,W_1$ are linear functions in $x$ and $t$ with real coefficients and $M_j,\alpha_k, \beta_k, \gamma_k, \theta_k \, (k =1\dots,6)$ are real constants (see Appendix for their explicit expressions). The above representations for $f$ and $g$ reveal that \eqref{first-order breather} is a breather solution to the Sasa-Sastuma equation \eqref{SS equation}.

%{\color{red} To illustrate this, by setting...
%the equation \eqref{SS equation} admits the solution}

\begin{figure*}[!ht]
  \centering
\subfigure[]{%
      \includegraphics[width=55mm]{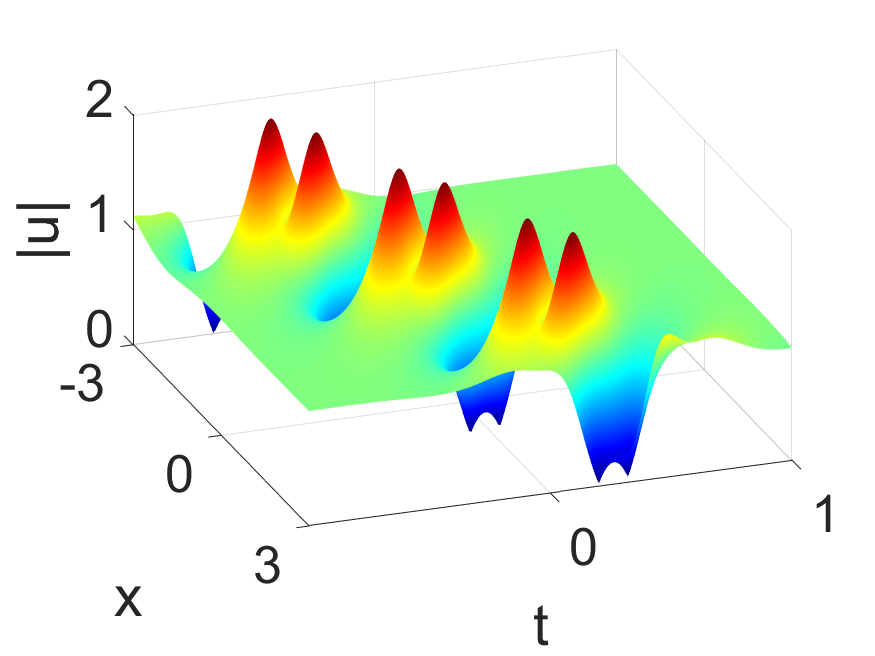}
      \label{1st order breather-1-1}}
\subfigure[]{%
      \includegraphics[width=55mm]{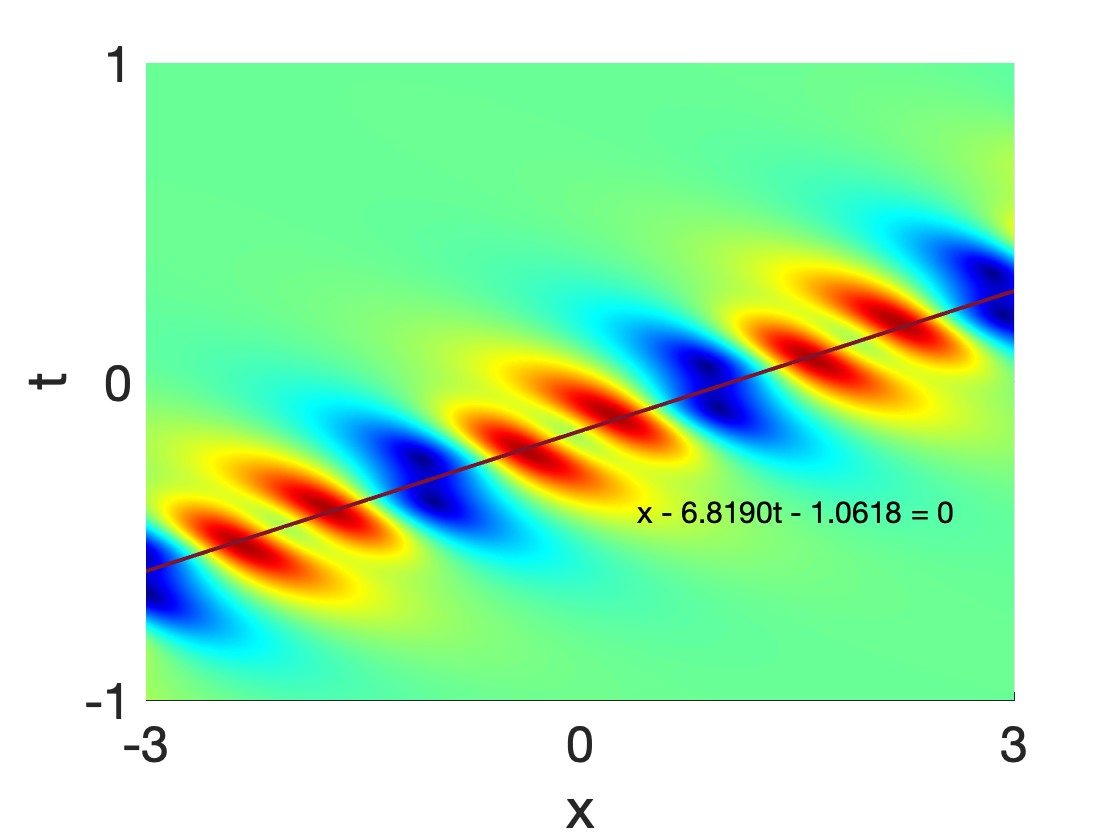}
      \label{1st order breather-1-2}}
\subfigure[]{%
      \includegraphics[width=55mm]{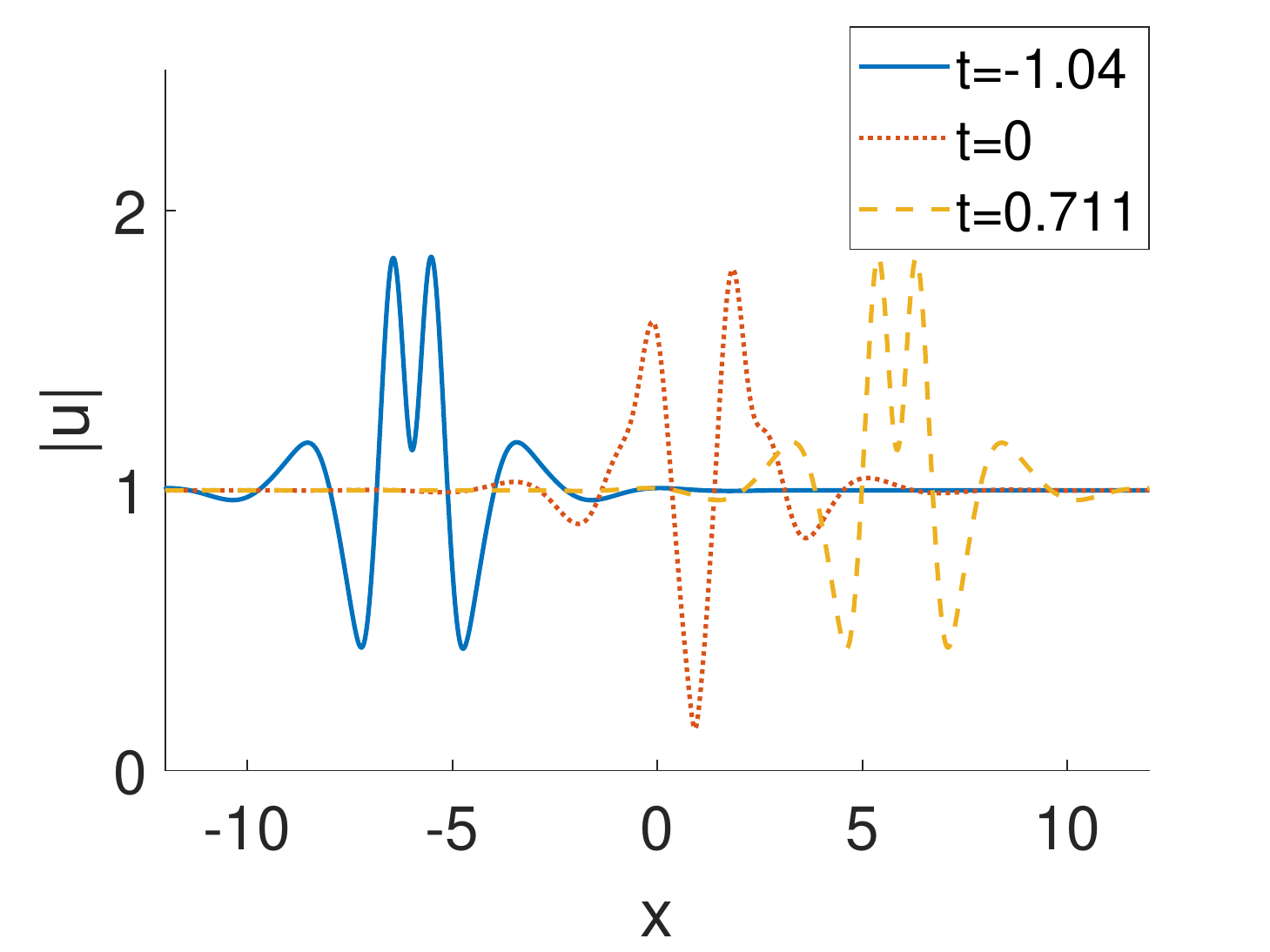}
      \label{1st order breather-1-3}}
\subfigure[]{%
      \includegraphics[width=55mm]{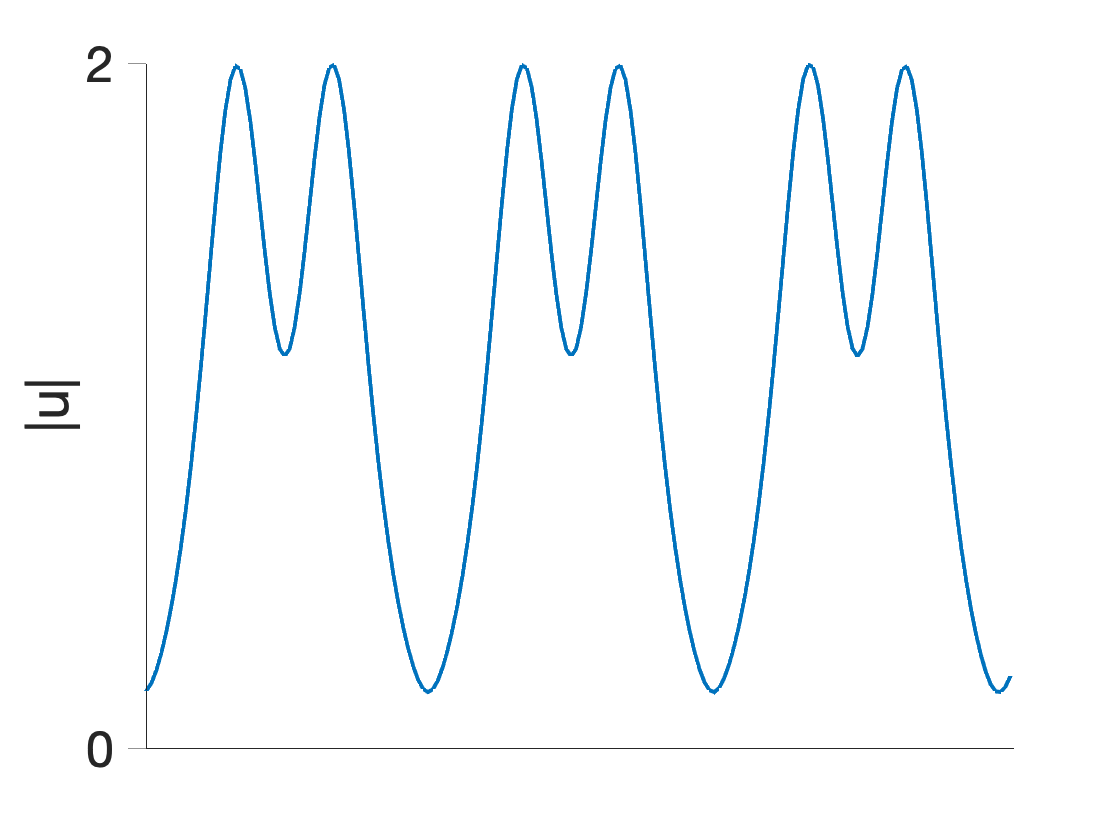}
      \label{1st order breather-1-4}}
  \caption{(Color online) A first-order breather solution with parameter values \(c=-1,\kappa=-1/2,\xi_{11,0}=\xi_{12,0}=0\), \(p_{11}=1+2\mathbf{i}\), $ p_{12} =  4/29-9\mathbf{i}/58$. (b) is the corresponding density plot of (a), (c) corresponds to the time evolution of (a) and (d) is  the intersection between the plane $ x - 6.819t - 1.0618=0$ and the breather. }
  \label{1st order breather-1}
\end{figure*}

In contrast with many integrable equations, a remarkable feature displayed by the Sasa-Satsuma equation \eqref{SS equation} is that it possesses double-hump one soliton solutions \cite{Sasa1991Satsuma}. Interestingly, this property can also be discovered in the breather solutions. This type of breather solution for parameters
\begin{equation*}
  c= - 1, \quad \kappa = -1/2, \quad p_{11} = 1+2\mathbf{i}, \quad p_{12} = \frac{4}{29}+\frac{9}{58} \mathbf{i}, \quad \xi_{11,0}=0, \quad \xi_{12,0}=0
\end{equation*}
is depicted in Figure \ref{1st order breather-1-1}. It is clear that this first-order breather contains two local maxima and three local minima in each period, where one local minimum is much bigger than the other two and located between two local maxima while the other two local minima are located on the same side of the local maxima. To be more precise, this breather reaches its peaks  at $(x,t) \approx (1.6100,   0.0700), (2.2000,    0.1850) $, and a trough at $(x,t) \approx (1.9150,    0.1250) $. Numerical computations indicate that its period is approximately 2.00112 and the local minima between two local maxima are located on the line  $L: x \approx 6.819t + 1.0618$ (see Figure \ref{1st order breather-1-2}). As displayed in Figure \ref{1st order breather-1-4}, taking the intersection of the line, the breather produces a double-hump periodic wave. Therefore, this breather may serve as a counterpart of the double-hump one soliton of the Sasa-Satsuma equation.

\begin{figure*}[!ht]
  \centering
\subfigure[]{%
      \includegraphics[width=55mm]{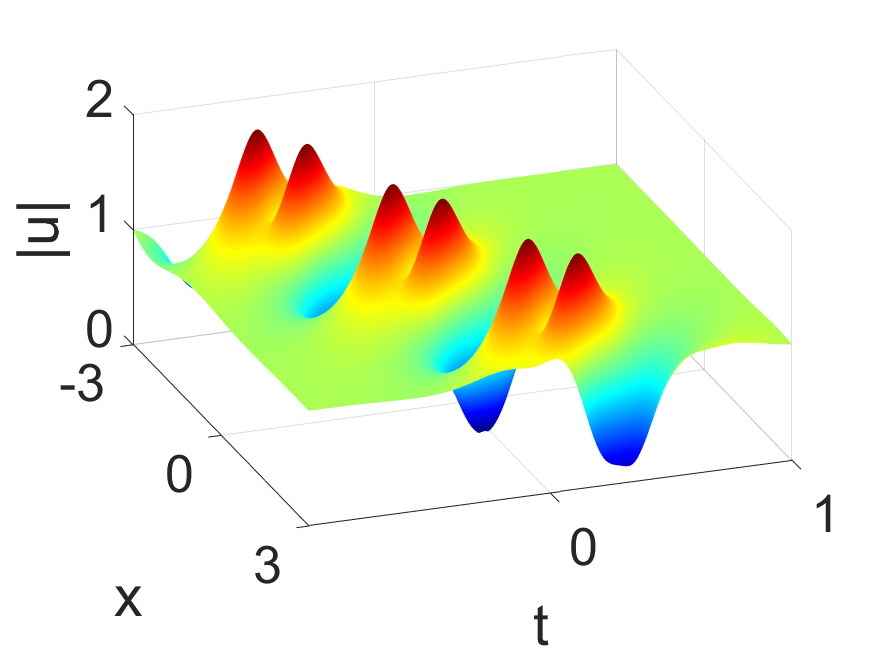}
      \label{1st order breather-2-1}}
\subfigure[]{%
      \includegraphics[width=55mm]{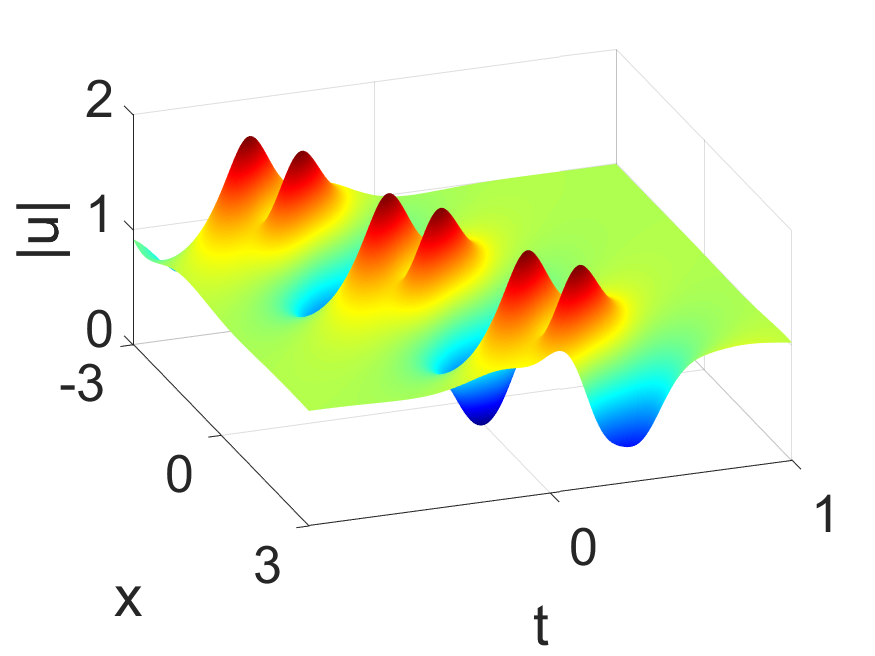}
      \label{1st order breather-2-2}}
\subfigure[]{%
      \includegraphics[width=55mm]{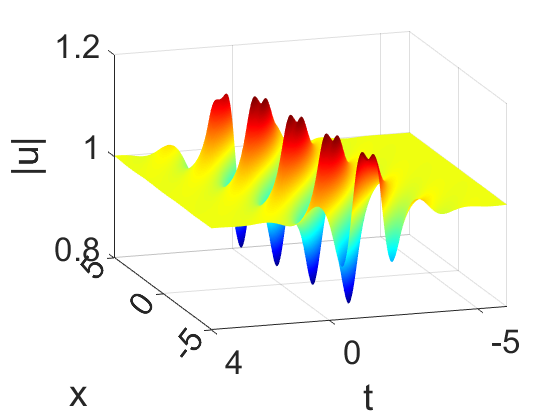}
      \label{1st order breather-2-3}}
\subfigure[]{%
      \includegraphics[width=55mm]{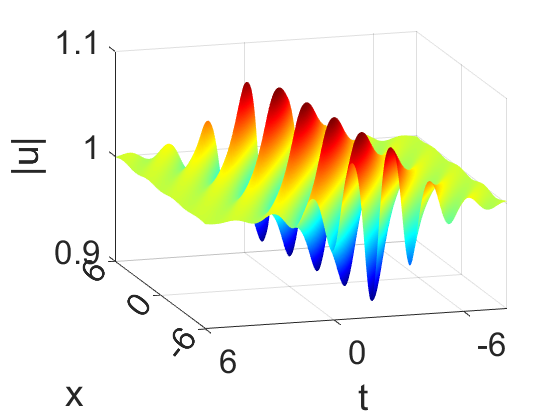}
      \label{1st order breather-2-4}}
  \caption{(Color online) First-order breather solutions with parameter values \(c=-1,\kappa=-1/2,\xi_{11,0}=\xi_{12,0}=0\) and (a) \(p_{11}=0.87+2\mathbf{i}\), (b) \(p_{11}=0.8+2\mathbf{i}\), (c) \(p_{11}=0.2+2\mathbf{i}\), (d) \(p_{11}=0.1+2\mathbf{i}\),  where \(p_{12}\) is given by \eqref{p12}.} %(c) and (d) are the corresponding density plots of (a) and (b), respectively.}
  \label{1st order breather-2}
\end{figure*}

According to Remark \ref{number of free parameters}, the solutions \eqref{first-order breather} contain seven free real parameters. Varying these parameters will excite  various %intricate
interesting wave profiles of the breather solutions. To illustrate this, we fix the parameter values
\begin{eqnarray*}
   c= - 1, \quad \kappa = -1/2, \quad \Im{p_{11}} = 2, \quad \xi_{11,0}=0, \quad \xi_{12,0}=0
\end{eqnarray*}
and let $p=\Re{p_{11}}$ be free. In addition, we choose (see Remark \ref{p_i2})
\begin{equation}\label{p12}
p_{12} = \frac{4 \kappa ^2 p_{i1} -  \mathbf{i} \kappa  \left(p_{i1}^2-3 \kappa ^2\right)}{\kappa ^2+p_{i1}^2}.
\end{equation}
Then the wave profiles of the breather solutions can exhibit an intriguing sequence of transitions by altering the values of $p$. Geometrically these wave profiles can be defined as $(m,n)$-type, where $m$ and $n$ represent the numbers of local maxima and minima in one period respectively. If we start from $p=1$, then previous discussions imply that it corresponds to a $(2,3)$-type breather (see Figure \ref{1st order breather-1}). Subsequently, the two smaller local minima will approach each other and merge into a single minimum by changing $p$ and hence the wave profile becomes $(2,2)$-type (see Figures \ref{1st order breather-2-1} and \ref{1st order breather-2-2}). On further changing $p$, the local minimum located between two local maxima is converted to a saddle point and the breather turns into $(2,1)$-type (see Figure \ref{1st order breather-2-3}). This is followed by $(1,1)$-type breather (see Figure \ref{1st order breather-2-4}) with the decrease of $p$ after two local maxima coalesce into a single maximum.

In the above process, the sign of $p$ is positive. Interestingly, similar behaviours can be observed as well for negative $p$. In this case, the wave profiles will traverse the three types of $(1,2), (2,2)$ and $(3,2)$ by varying $p$ (see Figure \ref{1st order breather-3}).

%Note that when we fix the parameter values of $c,\kappa$ and $p_{11}$, the equation \eqref{parameter constraint1} yields two choices for $p_{12}$. Thus, distinct configurations of breather profiles for the same input parameters are possible. The first possible configuration is depicted in Figure ??, while the second
%complex root of the equation \eqref{parameter constraint1}  gives $p_{12} = ???$, leading to a completely different wave profile.

\begin{figure*}[!ht]
  \centering
\subfigure[]{%
      \includegraphics[width=43mm]{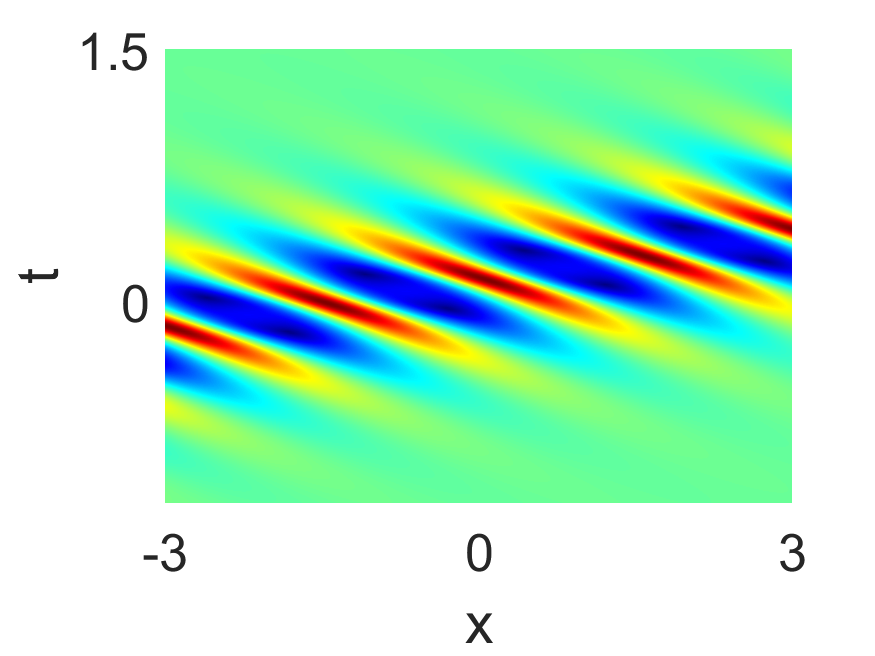}
      \label{1st order breather-3-1}}
\subfigure[]{%
      \includegraphics[width=43mm]{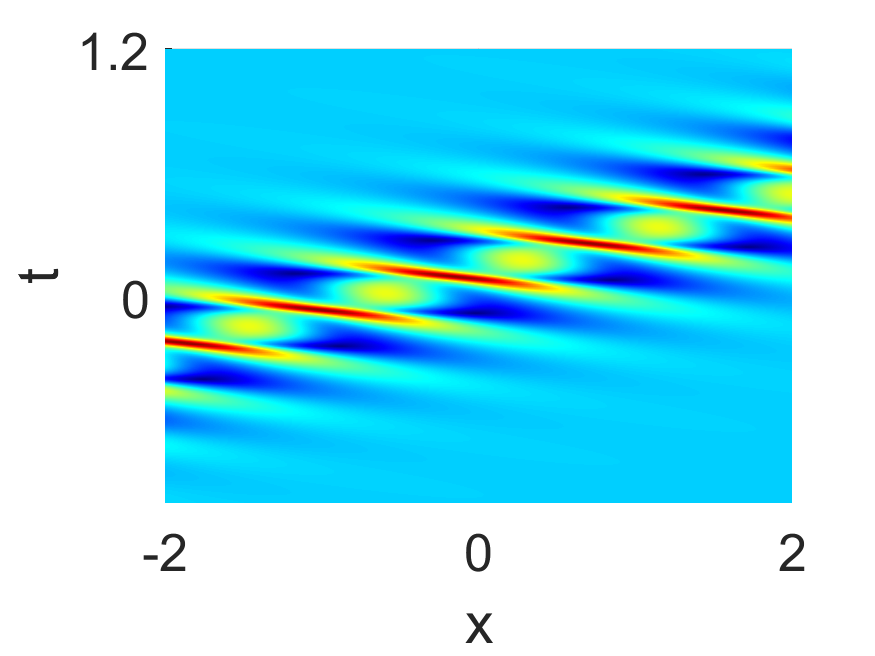}
      \label{1st order breather-3-2}}
\subfigure[]{%
      \includegraphics[width=43mm]{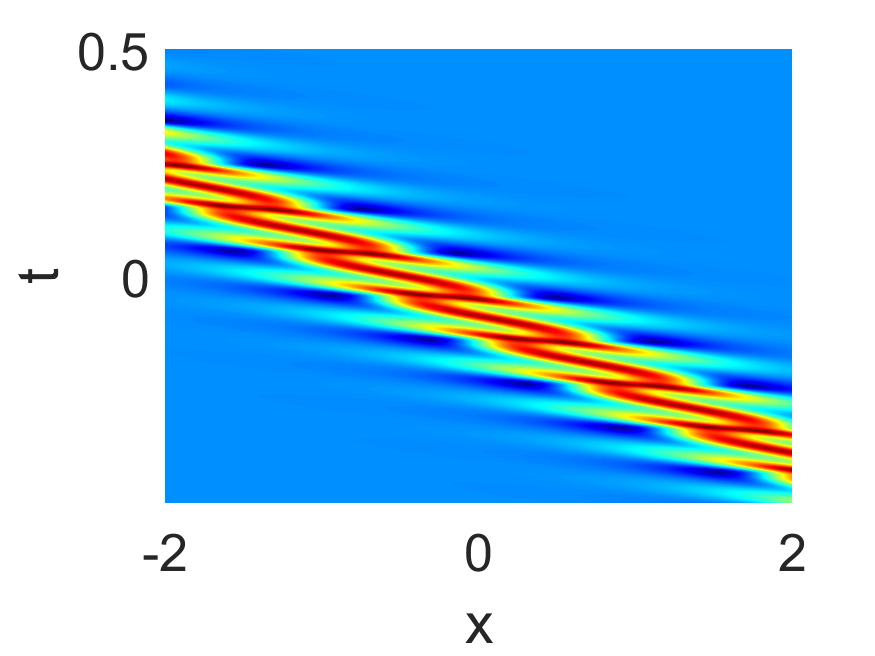}
      \label{1st order breather-3-3}}
  \caption{(Color online) First-order breather solutions with parameter values \(c=-1,\kappa=-1/2,\xi_{11,0}=\xi_{12,0}=0\) and (a) \(p_{11}=-0.6+2\mathbf{i}\), (b) \(p_{11}=-1.6+2\mathbf{i}\), (c) \(p_{11}=-3.5+2\mathbf{i}\), where \(p_{12}\) is given by \eqref{p12}. }
  \label{1st order breather-3}
\end{figure*}
%{\color{red} (2,3) and (3,2) (Figs 1, 6) can be improved}.

Note that when we fix the parameter values of $c,\kappa$ and $p_{11}$, the equation \eqref{parameter constraint1} yields two choices for $p_{12}$. Thus, distinct configurations of breather profiles for the same input parameters are possible. The first possible configuration is depicted in Figure \ref{1st order breather-1-1}, while the second complex root of the equation \eqref{parameter constraint1} gives
$p_{12}=4/13-25/26\mathbf{i}$,
%\begin{equation}\label{p12c}
%  p_{12} = \frac{4 \kappa ^2 p_{i1} +  \mathbf{i} \kappa  \left(p_{i1}^2-3 \kappa ^2\right)}{\kappa ^2+p_{i1}^2},
%\end{equation}
leading to a completely different wave profile (see Figure \ref{1st order breather-4}).

\subsection{Higher-order breather solutions}

Second-order breather solutions to the equation \eqref{SS equation} correspond to $N = 2$ in \eqref{SS solution_tau fucntion}. In this circumstance, the functions $\tau_k \, (k=0,1)$ could be obtained from \eqref{SS solution_tau fucntion} as
%\begin{equation*}
%  \tau_k = \left|  m_{ij}^{(k)}  \right|_{1\leq i,j \leq 4}
%\end{equation*}
\begin{eqnarray*}
%\begin{aligned}
\tau_0 = \left|
\begin{array}{cccc}
   m_{11}^{(0)} &m_{12}^{(0)} & m_{13}^{(0)} &m_{14}^{(0)}   \\
 m_{21}^{(0)} &m_{22}^{(0)} & m_{23}^{(0)} &m_{24}^{(0)}   \\
   m_{31}^{(0)} &m_{32}^{(0)} & m_{33}^{(0)} &m_{34}^{(0)}   \\
 m_{41}^{(0)} &m_{42}^{(0)} & m_{43}^{(0)} &m_{44}^{(0)}
\end{array}
\right|
,
\quad
\tau_1 = \left|
\begin{array}{cccc}
   m_{11}^{(1)} &m_{12}^{(1)} & m_{13}^{(1)} &m_{14}^{(1)}   \\
 m_{21}^{(1)} &m_{22}^{(1)} & m_{23}^{(1)} &m_{24}^{(1)}   \\
   m_{31}^{(1)} &m_{32}^{(1)} & m_{33}^{(1)} &m_{34}^{(1)}   \\
 m_{41}^{(1)} &m_{42}^{(1)} & m_{43}^{(1)} &m_{44}^{(1)}
\end{array}
\right|,
%\end{aligned}
\end{eqnarray*}
with matrix entries
\begin{equation*}%\label{}
  m_{ij}^{(k)} = \sum\limits_{m,n=1}^{2} \frac{1}{p_{im} + p_{jn}} \left(\dfrac{a-p_{1 m}}{a+p_{1 n}}\right)^k e^{\xi_{im} + \xi_{jn}}, \quad i,j=1, \dots, 4,
\end{equation*}
where $a = \mathbf{i} \kappa$ is purely imaginary, $ \xi_{im}=p_{im} x +p_{im}^{3} t+ \xi_{im,0} \, (m=1, 2)$,
and the complex parameters $ \xi_{im,0}, p_{im} $  satisfy the relations  \eqref{parameter constraint1} and \eqref{parameter constraint2-1}. Similar to the first-order breather solutions, the second-order breather solutions can also be expressed in terms of trigonometric functions and hyperbolic functions. Since the expressions are very complicated, we omit their explicit forms.
\begin{figure*}[!ht]
  \centering
\subfigure[]{%
      \includegraphics[width=40mm]{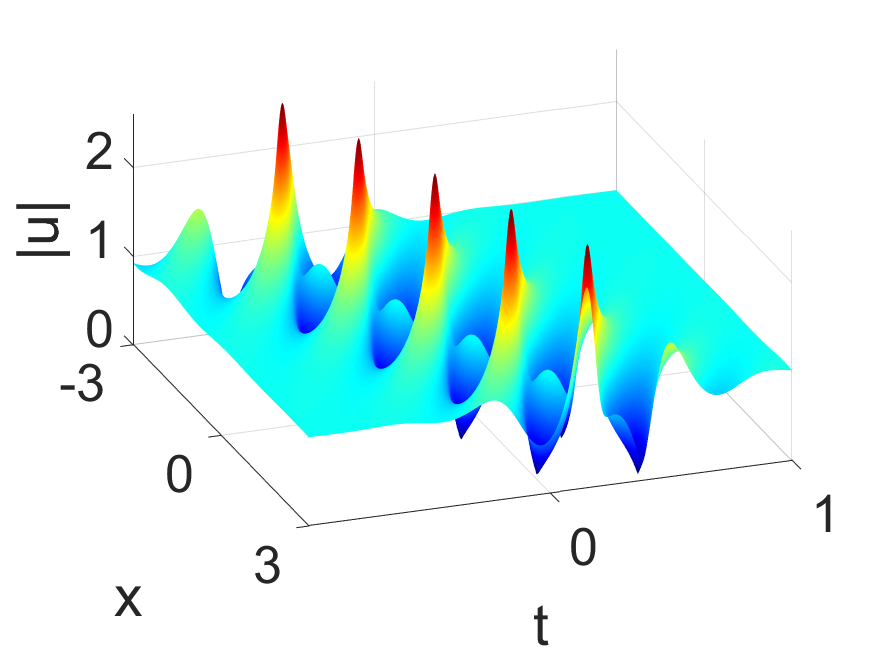}
      \label{1st order breather-4-1}}
\subfigure[]{%
      \includegraphics[width=40mm]{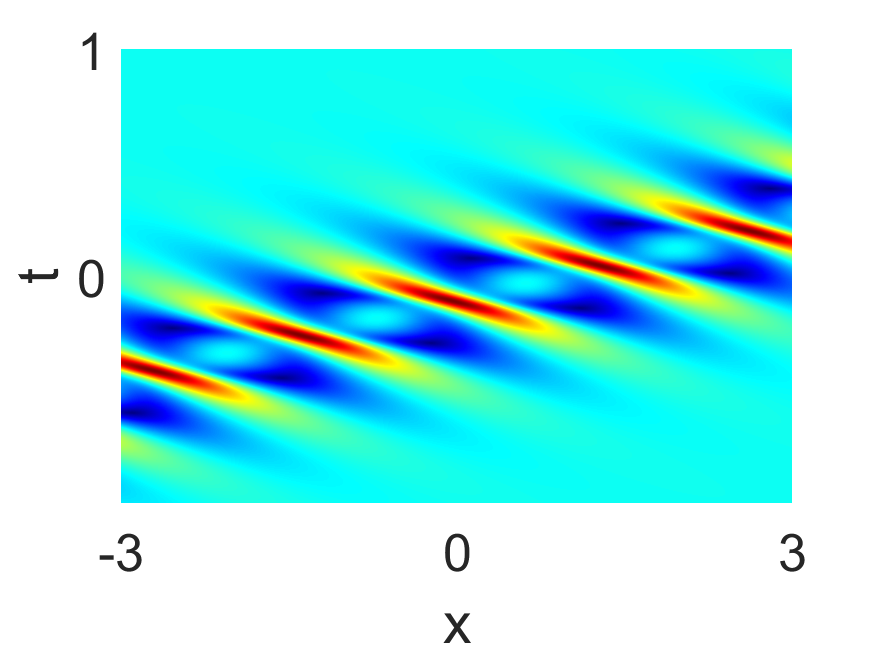}
      \label{1st order breather-4-2}}
  \caption{(Color online) First-order breather solutions with parameter values \(c=-1,\kappa=-1/2,\xi_{11,0}=\xi_{12,0}=0\) and (a) \(p_{11}=1+2\mathbf{i},p_{12}=4/13-25/26\mathbf{i}\). (b) is the corresponding density plots of (a). }
  \label{1st order breather-4}
\end{figure*}
As pointed in Remark \ref{number of free parameters}, second-order breather solutions contain the free parameters $\kappa, p_{i1} $ and $\xi_{im, 0} \, (i, m =1,2)$, where $\kappa$ is real and $p_{i1}, \xi_{im, 0}$ are complex. A variety of fascinating wave profiles can be depicted for different choices of parameter values. Since second-order breathers describe the interactions between two first-order breathers, each of them can be classified into $(m_1,n_1)$-$(m_2,n_2)$-type if it comprises two first-order breathers that are $(m_1,n_1)$-type and $(m_2,n_2)$-type respectively. In Section \ref{Dynamics of breather solutions}\ref{First-order breather solution}, six types of first-order breathers have been illustrated, and hence they give rise to 21 types of second-order breathers. To demonstrate this, we take the parameters
\begin{eqnarray*}
 && c= - 1, \quad \kappa = -1/2, \quad p_{11} = 0.95+1.65\mathbf{i}, \quad p_{12} = 0.8+2\mathbf{i}, \quad p_{21} = \frac{38}{221}  + \frac{49}{442}\mathbf{i},\\
 && p_{22} = \dfrac{80}{689} + \dfrac{189}{1378}\mathbf{i}, \quad \xi_{11,0}= 0, \quad \xi_{12,0}= 0, \quad \xi_{21,0}= 0, \quad \xi_{22,0}=0.
\end{eqnarray*}
As shown in Figure \ref{2nd order breather-1}, this corresponds to a $(2,2)$-$(2,3)$-type second-order breather. It can also be seen clearly that the two breathers pass through each other without any change of shape or velocity, and thus the collision between them is elastic. If we choose other parameter values, then we may obtain second-order breathers consisting of two first-order breathers that belong to distinct types (see Figure \ref{2nd order breather-2}) or the same type (see Figure \ref{2nd order breather-3}).
\begin{figure*}[!ht]
  \centering
\subfigure[]{%
      \includegraphics[width=43mm]{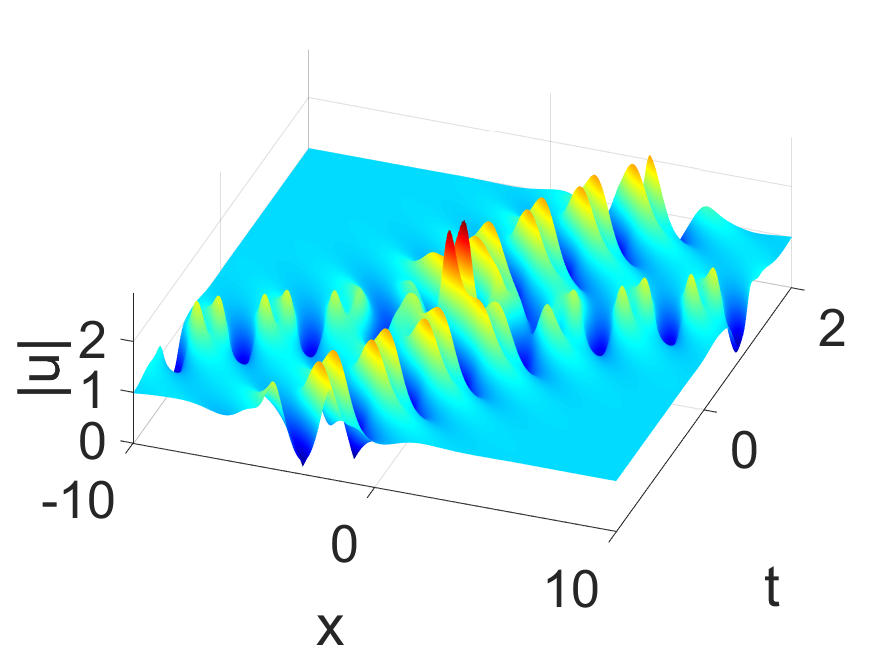}
      }
\subfigure[]{%
      \includegraphics[width=43mm]{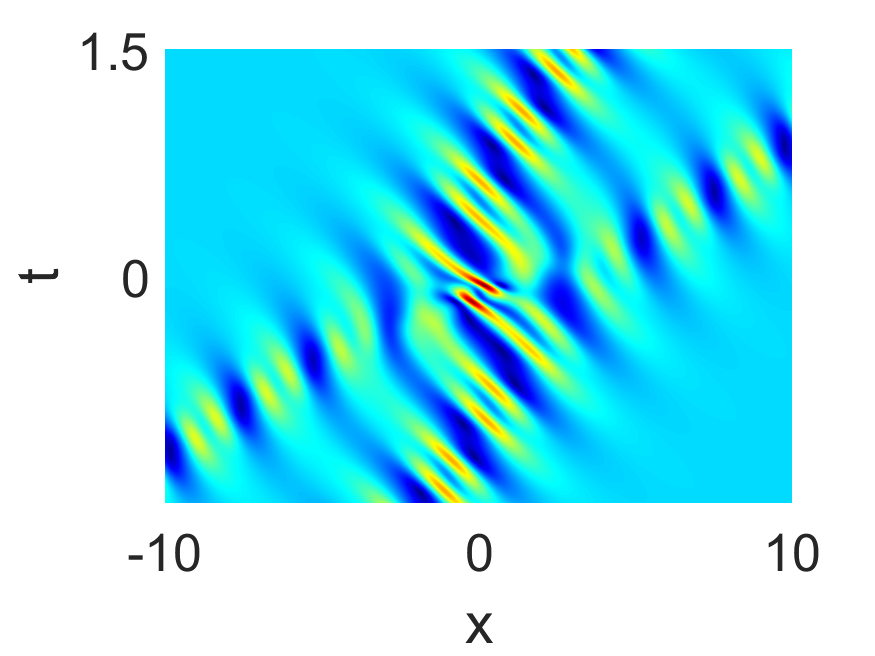}
      }
\subfigure[]{%
      \includegraphics[width=43mm]{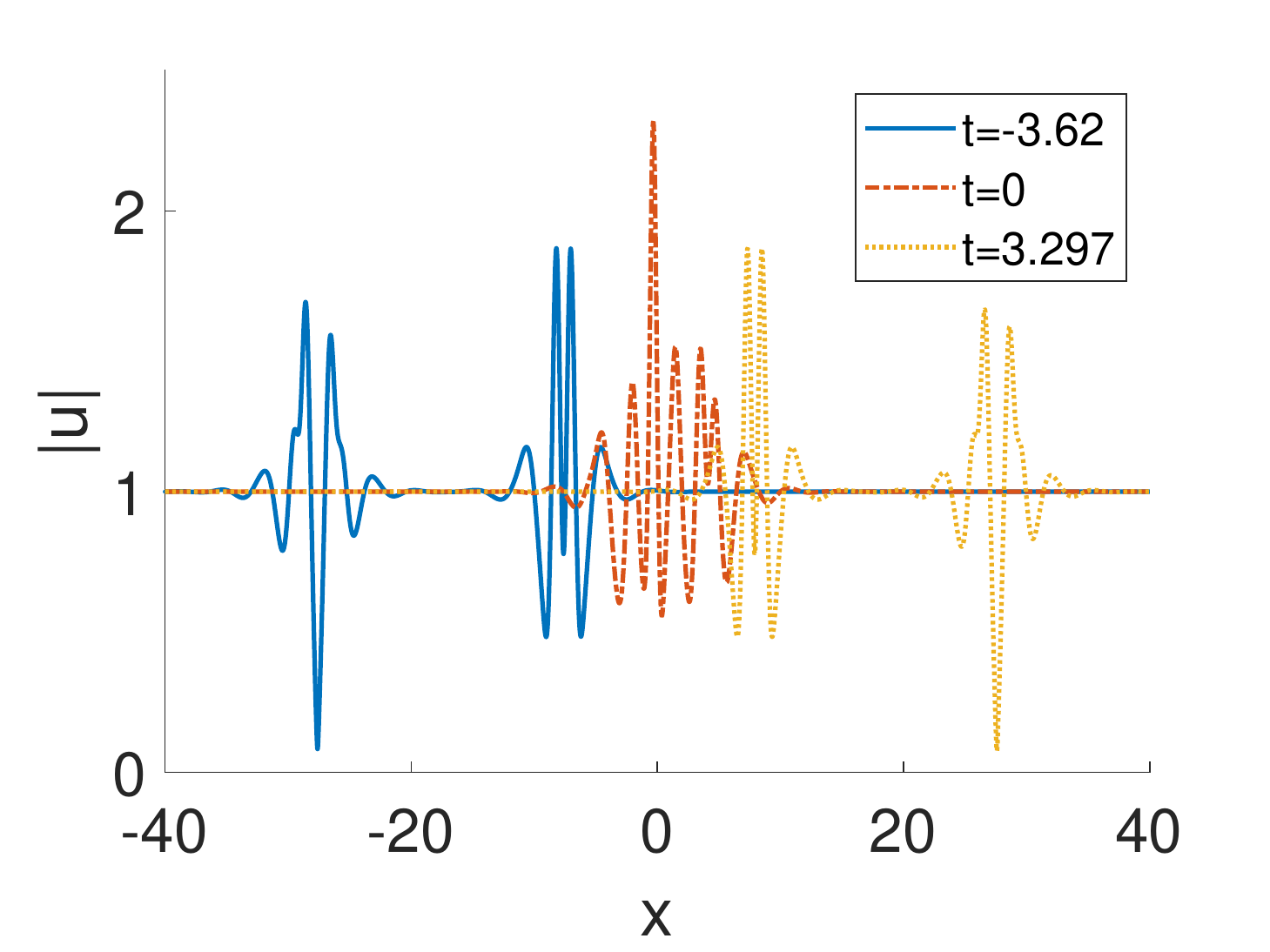}
      }
  \caption{(Color online) Second-order breather solutions with parameter values \(c=-1,\kappa=-1/2, p_{11}=0.95+1.65\mathbf{i}, p_{21}=0.8+2\mathbf{i}, \xi_{11,0}=\xi_{12,0}=0\), where \(p_{12}\) and \(p_{22}\) are given by \eqref{p12}. (b) is the corresponding density plot of (a), and (c) corresponds to the time evolution of (a). }
  \label{2nd order breather-1}
%  \caption{(Color online) Second-order breather solutions with parameter values \(c=-1,\kappa=-1/2,\xi_{11,0}=\xi_{12,0}=0\) and (a) \(p_{11}=0.8+3.2\mathbf{i}\), \(p_{21}=0.95+1.65\mathbf{i}\), (b) \(p_{11}=0.95+1.65\mathbf{i}\), \(p_{21}=0.8+2\mathbf{i}\), (c) \(p_{11}=1+1.7\mathbf{i}\), \(p_{21}=-0.65+2.5\mathbf{i}\), where \(p_{12}\) and \(p_{22}\) are given by \eqref{p12}. (d), (e) and (f) are the corresponding density plots of (a), (b) and (c), respectively. }
%  \label{2nd order breather-1}
\end{figure*}

\begin{figure*}[!ht]
  \centering
\subfigure[]{%
      \includegraphics[width=55mm]{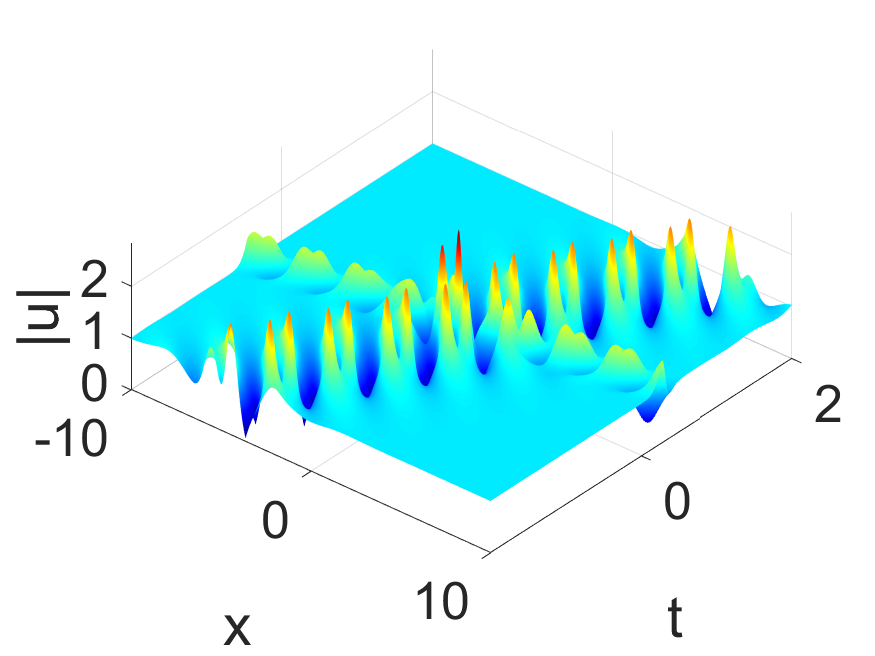}
      \label{2nd order breather-1-1}}
%\subfigure[]{%
%      \includegraphics[width=43mm,height=40mm]{./pic/11n.png}
%      \label{2nd order breather-1-2}}
\subfigure[]{%
      \includegraphics[width=55mm]{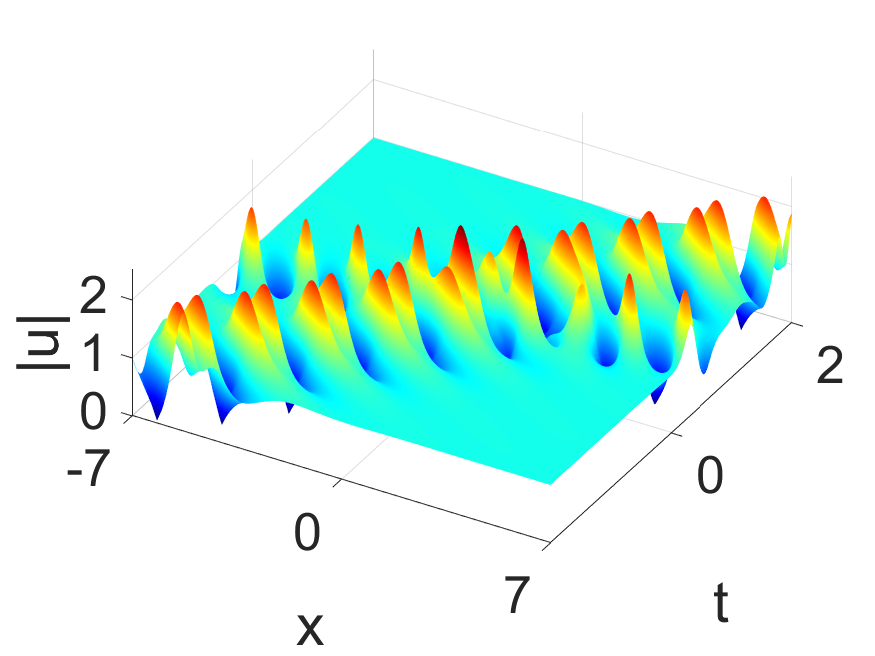}
      \label{2nd order breather-1-3}}
\subfigure[]{%
      \includegraphics[width=55mm]{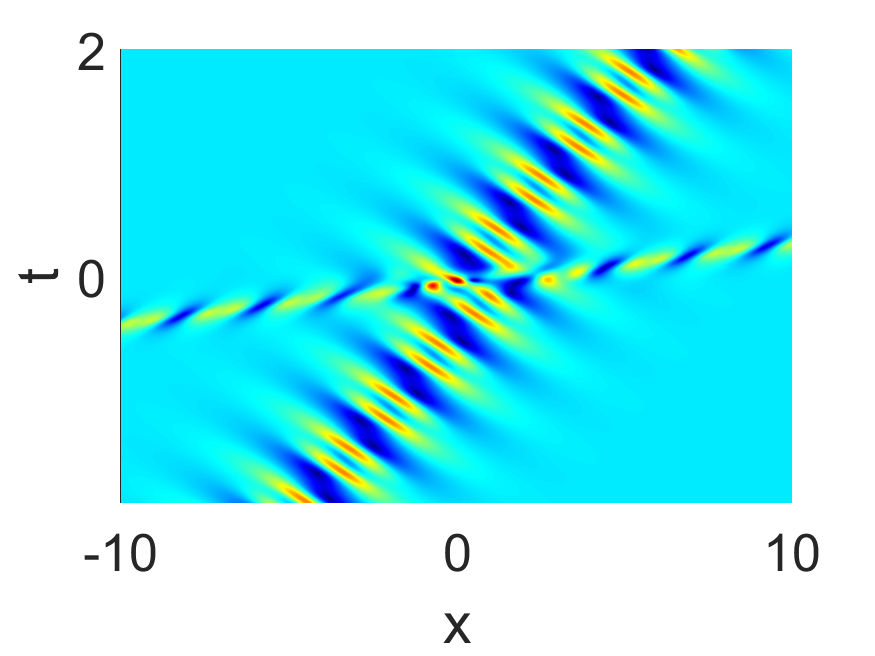}
      \label{2nd order breather-1-4}}
%\subfigure[]{%
%      \includegraphics[width=43mm]{./pic/11b.png}
%      \label{2nd order breather-1-5}}
\subfigure[]{%
      \includegraphics[width=55mm]{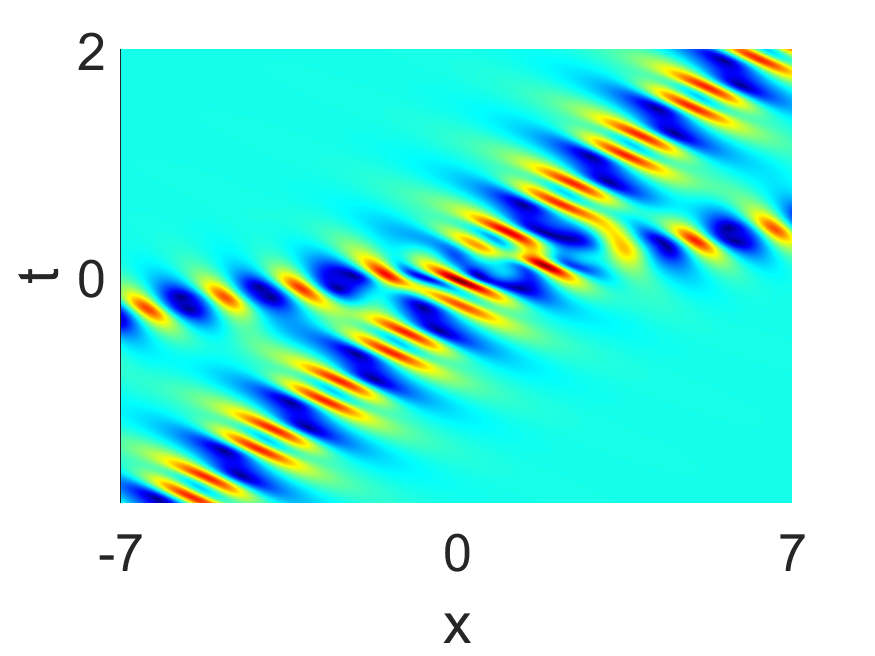}
      \label{2nd order breather-1-6}}
  \caption{(Color online) Second-order breather solutions with parameter values \(c=-1,\kappa=-1/2,\xi_{11,0}=\xi_{12,0}=0\) and (a) \(p_{11}=0.8+3.2\mathbf{i}\), \(p_{21}=0.95+1.65\mathbf{i}\),  (b) \(p_{11}=1+1.7\mathbf{i}\), \(p_{21}=-0.65+2.5\mathbf{i}\), where \(p_{12}\) and \(p_{22}\) are given by \eqref{p12}. (c)  and (d) are the corresponding density plots of (a)  and (b), respectively. }
  \label{2nd order breather-2}
\end{figure*}

\begin{figure*}[!ht]
  \centering
\subfigure[]{%
      \includegraphics[width=43mm,height=40mm]{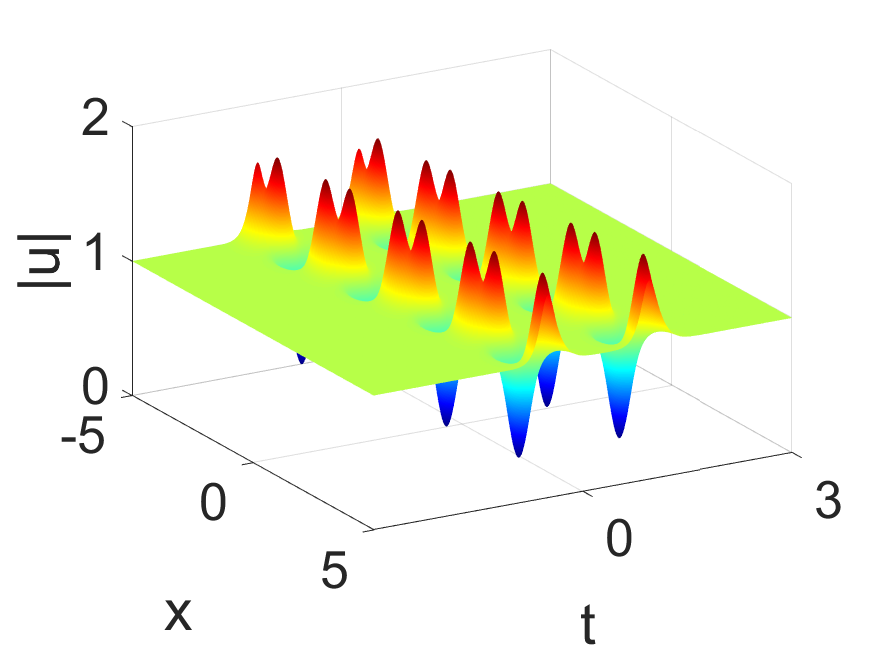}
      \label{2nd order breather-2-1}}
\subfigure[]{%
      \includegraphics[width=43mm,height=40mm]{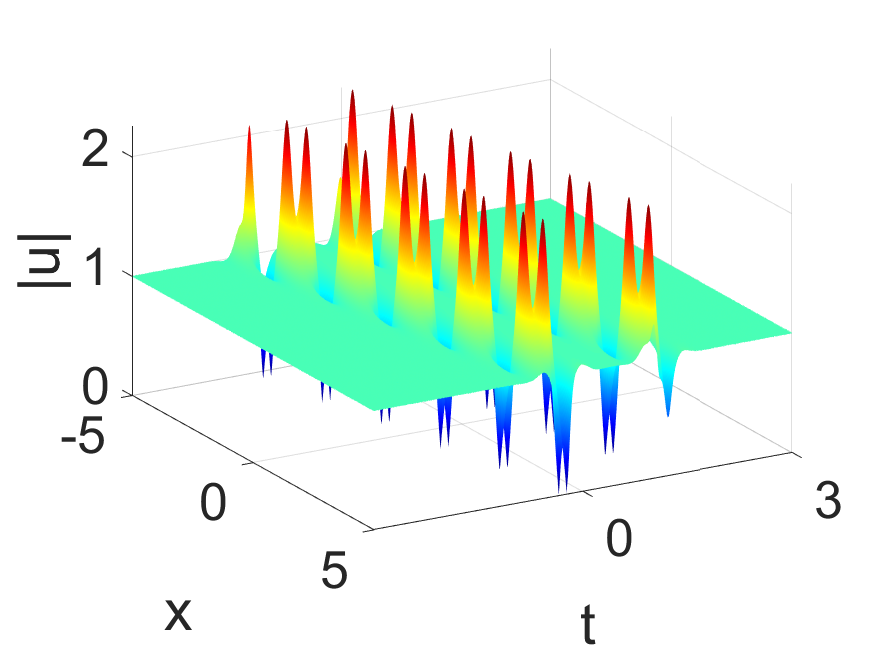}
      \label{2nd order breather-2-2}}
\subfigure[]{%
      \includegraphics[width=43mm,height=40mm]{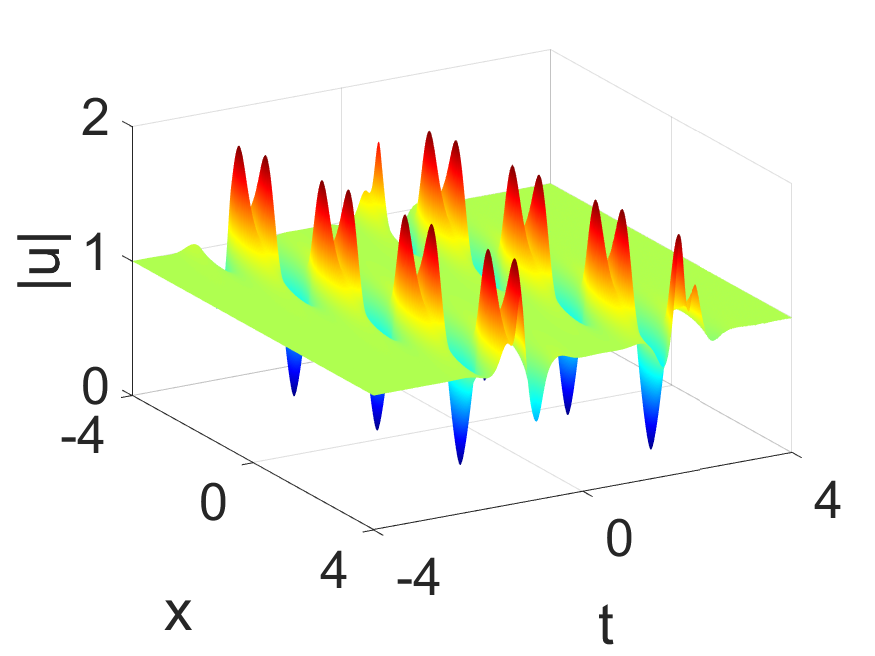}
      \label{2nd order breather-2-3}}
\subfigure[]{%
      \includegraphics[width=43mm]{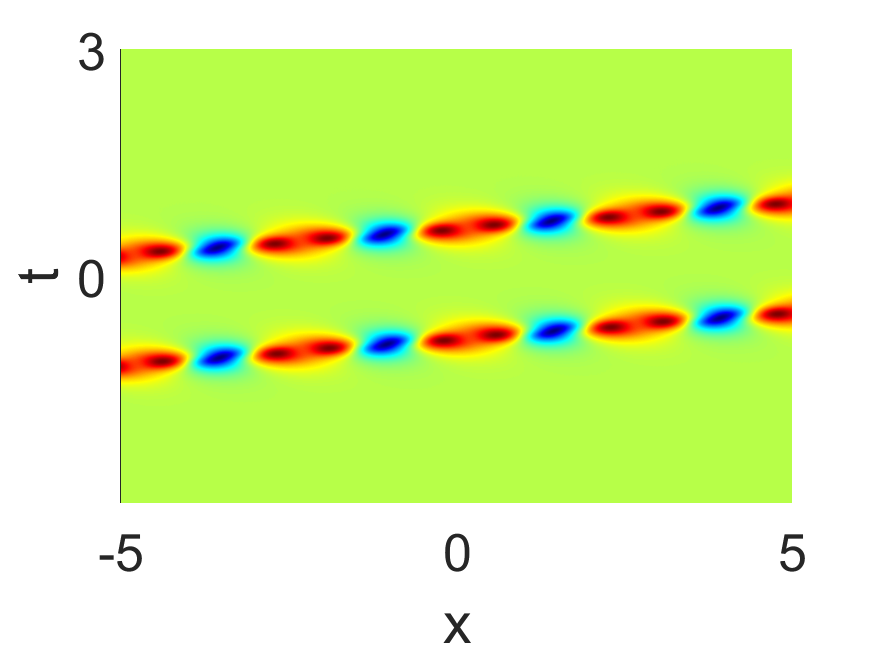}
      \label{2nd order breather-2-4}}
\subfigure[]{%
      \includegraphics[width=43mm]{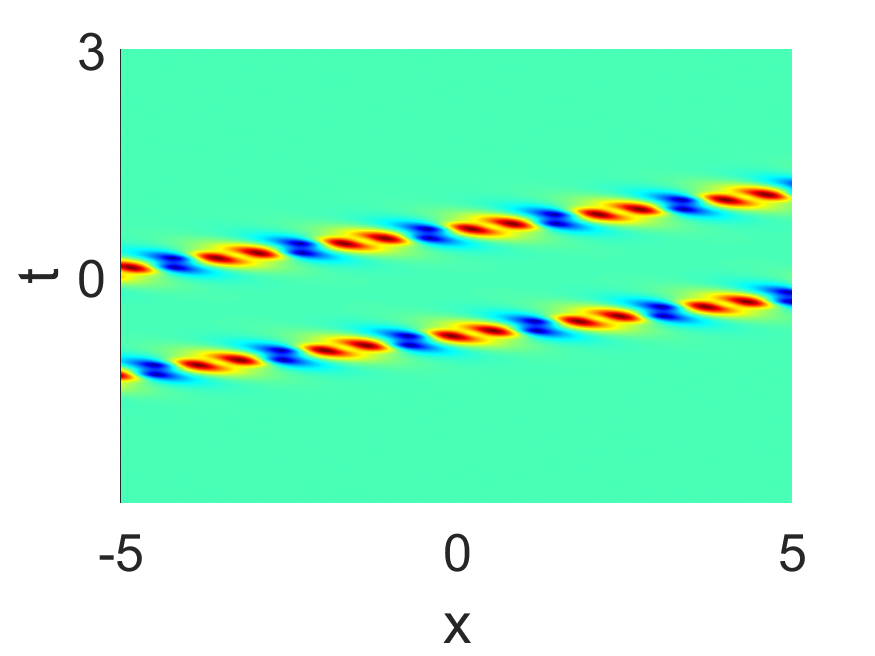}
      \label{2nd order breather-2-5}}
\subfigure[]{%
      \includegraphics[width=43mm]{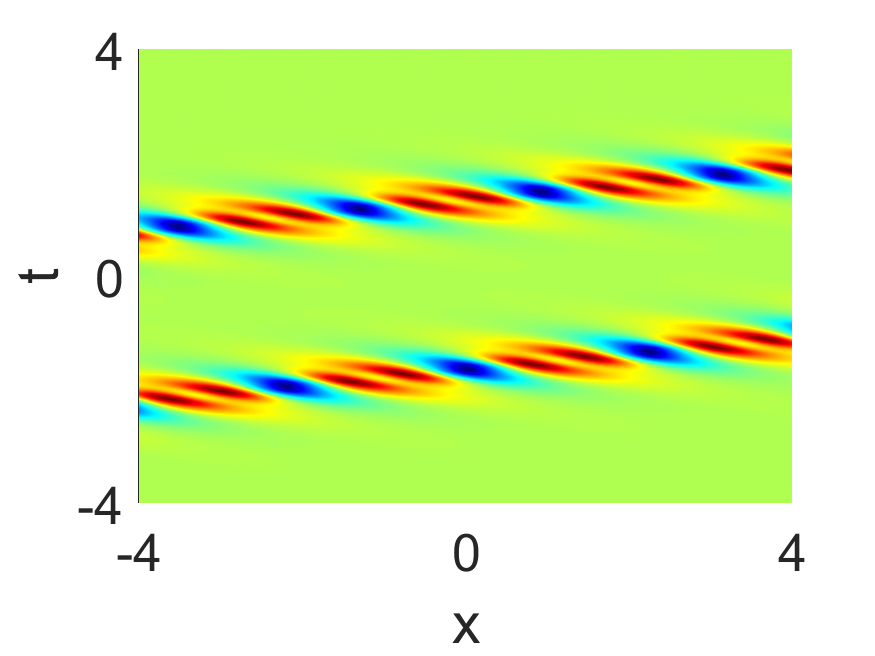}
      \label{2nd order breather-2-6}}
  \caption{(Color online) Second-order breather solutions with parameter values \(c=-1,\kappa=-1/2,\xi_{11,0}=1,\xi_{12,0}=0,\xi_{21,0}=\xi_{22,0}=0\) and (a) \(p_{11}=0.8+2.5\mathbf{i}\), \(p_{21}=0.8001+2.5\mathbf{i}\), (b) \(p_{11}=1.3+2.3\mathbf{i}\), \(p_{21}=1.3001+2.3\mathbf{i}\), (c) \(p_{11}=0.8+2\mathbf{i}\), \(p_{21}=0.8001+2\mathbf{i}\), where \(p_{12}\) and \(p_{22}\) are given by \eqref{p12}. (d), (e) and (f) are the corresponding density plots of (a), (b) and (c), respectively. }
  \label{2nd order breather-3}
\end{figure*}
Finally, we can obtain $N$th-order breather solutions to the equation \eqref{SS equation} from \eqref{SS solution_tau fucntion} by taking $N \geq 3$. In general, such solutions describe the superposition of $N$ first-order breathers. However, their explicit expressions are more complicated, so they will not be provided here. Instead, we only focus on the dynamical structures of third-order breather solutions ($N=3$), which consists of three first-order breathers. On the one hand, it is obvious that there are many more types of third-order breathers than second-order ones. On the other hand, third-order breathers exhibit more diverse collisions. As illustrated in Figure  \eqref{3rd order breather-1}, the three first-order breathers may interact with each other in pairs or collide simultaneously.
%In this circumstance, the functions $\tau_k \, (k=0,1)$ could be obtained from \eqref{SS solution_tau fucntion} as
\section{Conclusion}
In summary, we have derived general breather solutions to the SSE via the KP hierarchy reduction method. These solutions are expressed in terms of Gram-type
determinants through transforming a set of bilinear equations in the KP hierarchy into the bilinear forms of the SSE. Owing to the complexity of the SSE and multiple corresponding bilinear equations in the KP hierarchy, the intermediate computations are much more involved compared with most of the integrable equations that can be solved by the same method. Furthermore, in addition to the common obstructions that appear in the KP hierarchy reduction method, i.e., the dimension reduction and the complex conjugate reduction, another obstacle that we have dealt with is the symmetry reduction.

The dynamics of breathers have been investigated. For first-order breathers, six types were found totally and some of them were shown to possess a double-hump structure.
Interestingly, transitions among these first-order breathers can be achieved by changing the value of just one free real parameter in the solutions. In addition, various configurations of second- and third-order breathers have been illustrated. In particular, elastic collisions of second-order breathers were observed.

\begin{figure*}[!ht]
  \centering
\subfigure[]{%
      \includegraphics[width=43mm]{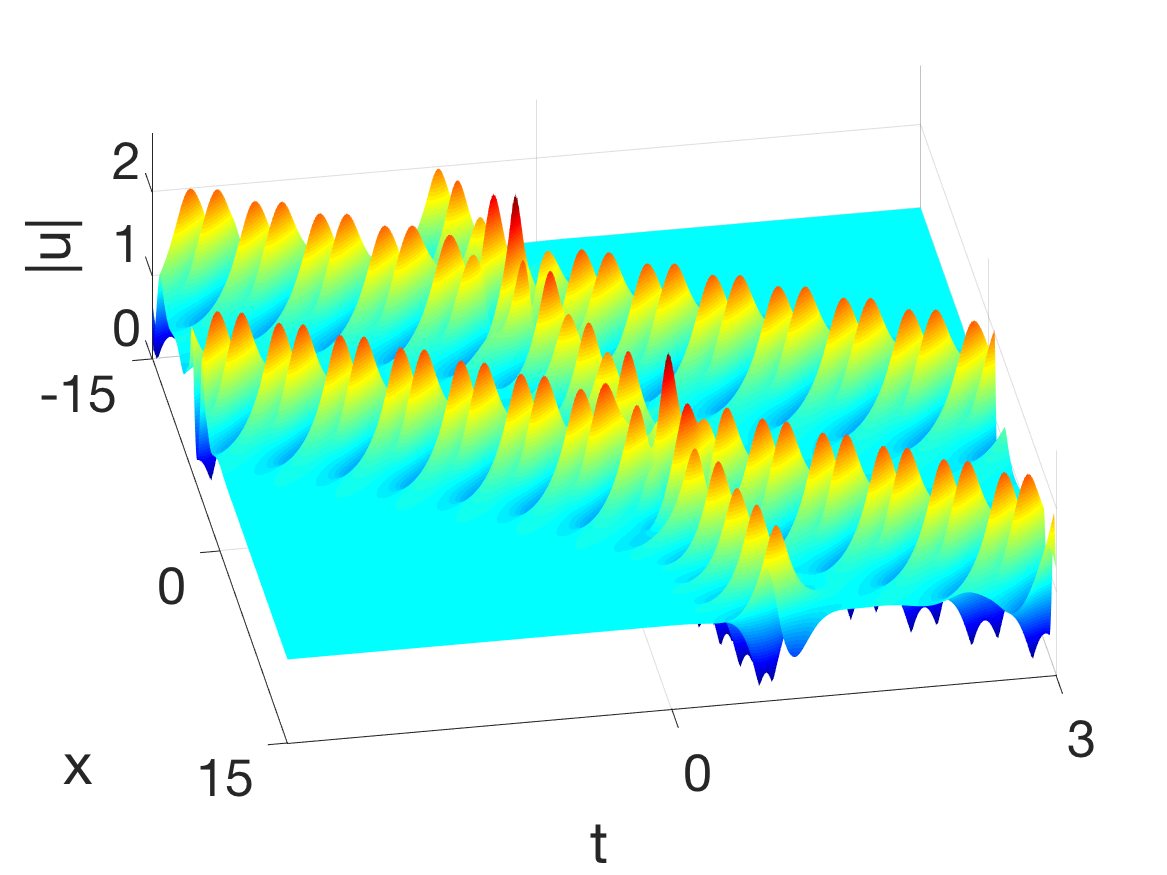}
      \label{3rd order breather-2-1}}
\subfigure[]{%
      \includegraphics[width=43mm]{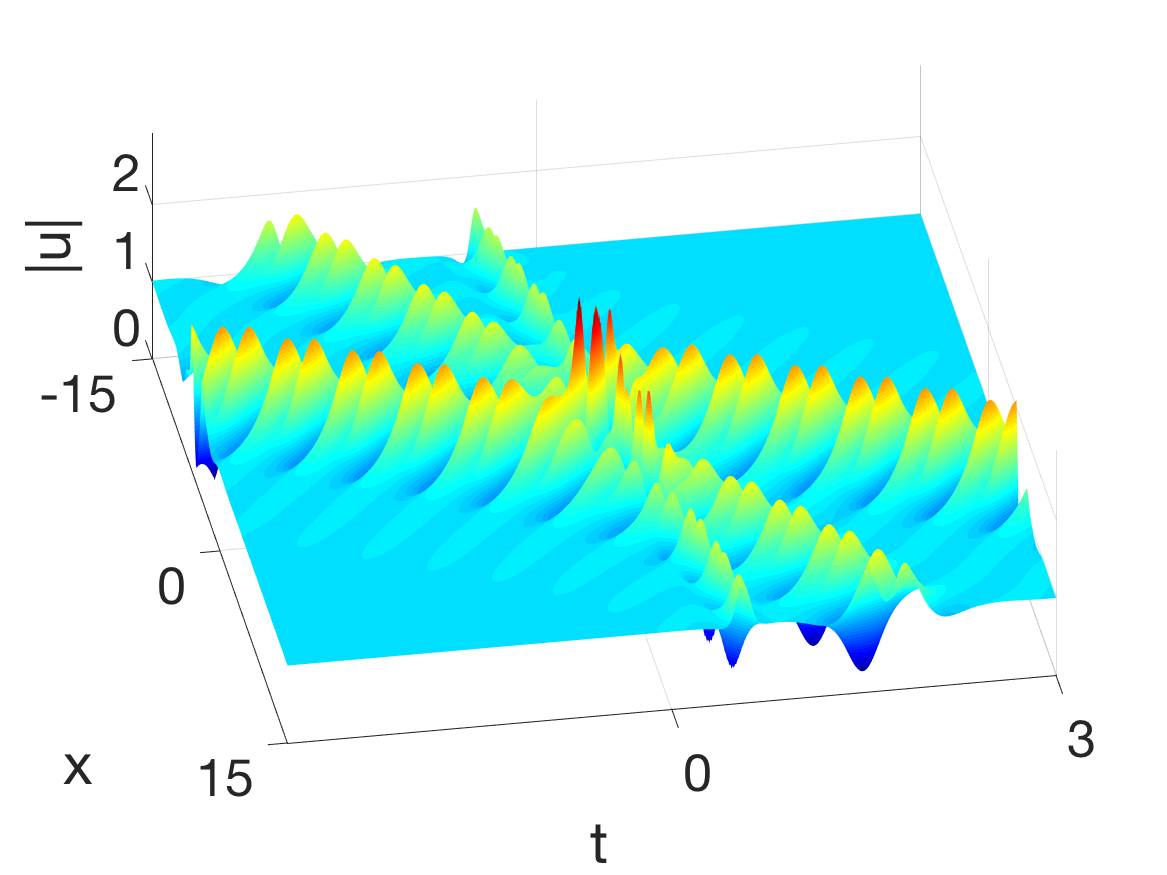}
      \label{3rd order breather-2-2}}
\subfigure[]{%
      \includegraphics[width=43mm]{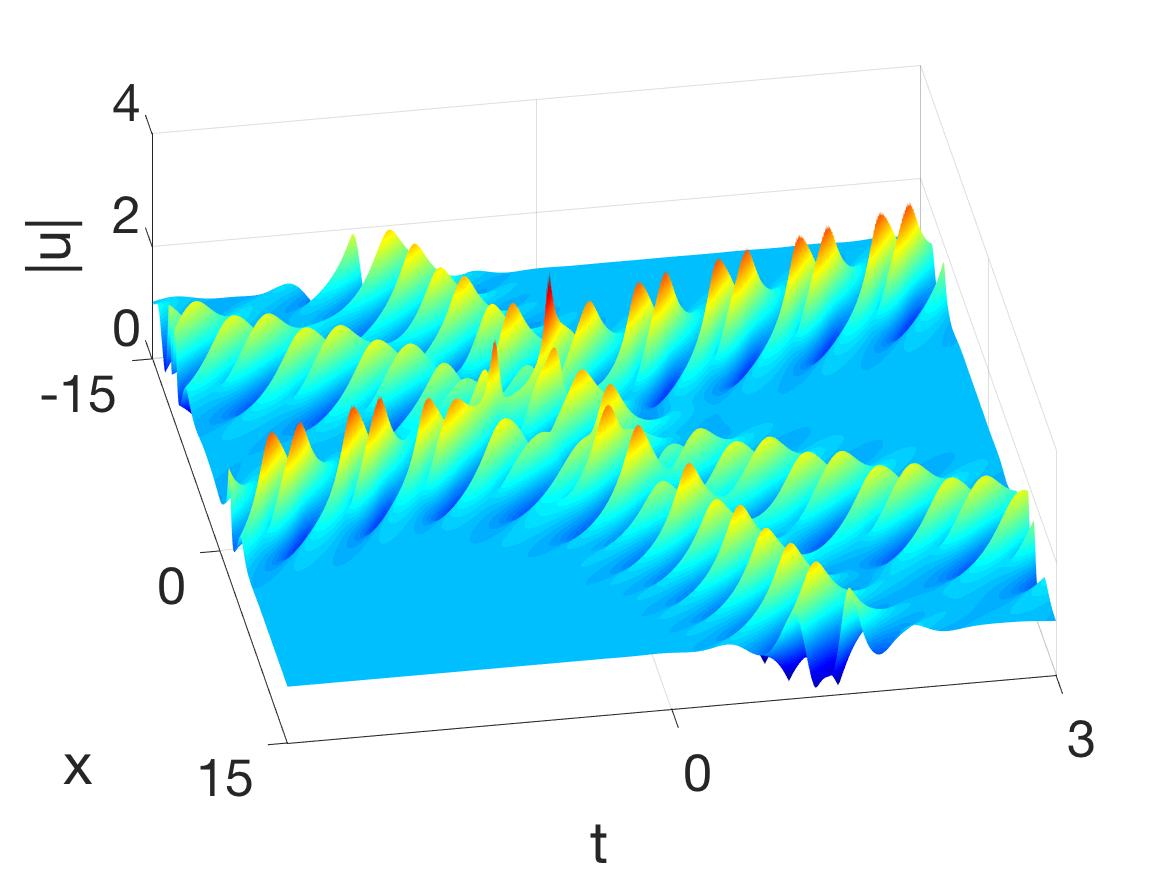}
      \label{3rd order breather-2-3}}
\subfigure[]{%
      \includegraphics[width=43mm]{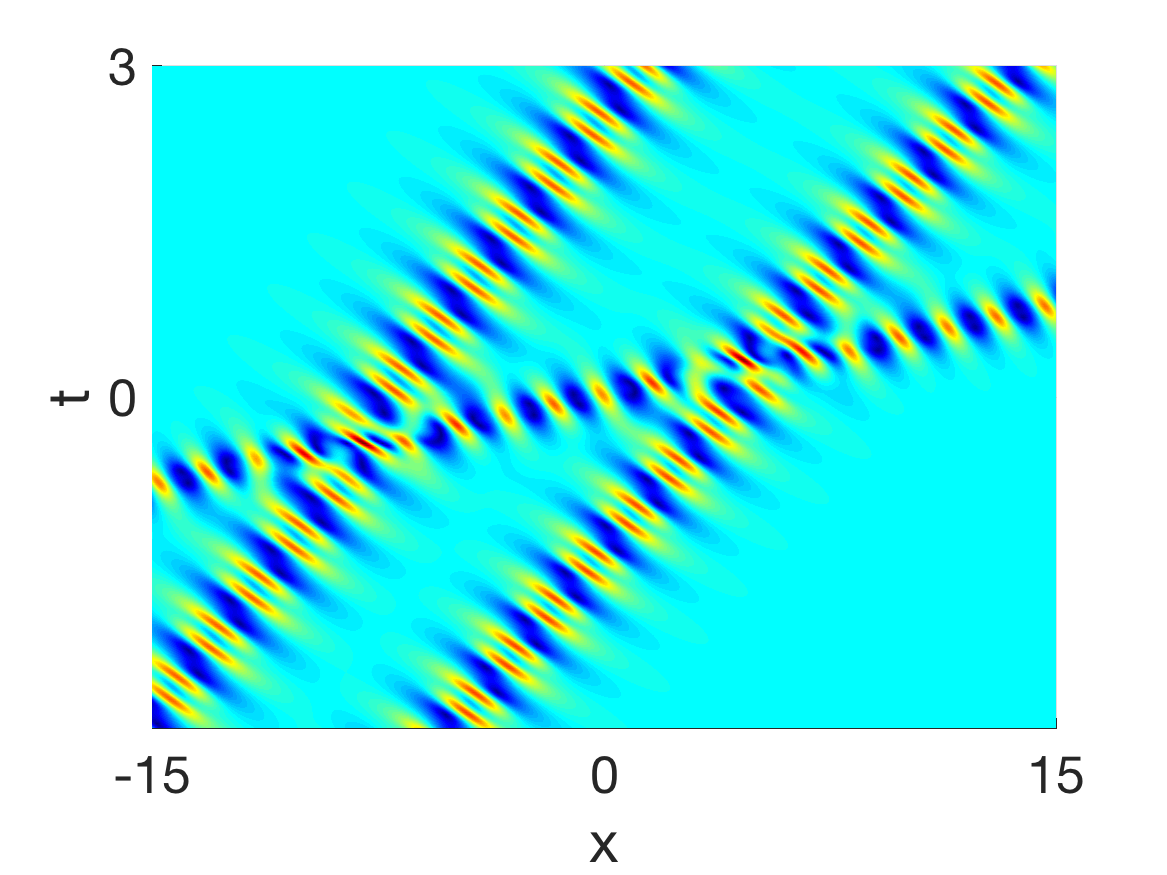}
      \label{3rd order breather-2-4}}
\subfigure[]{%
      \includegraphics[width=43mm]{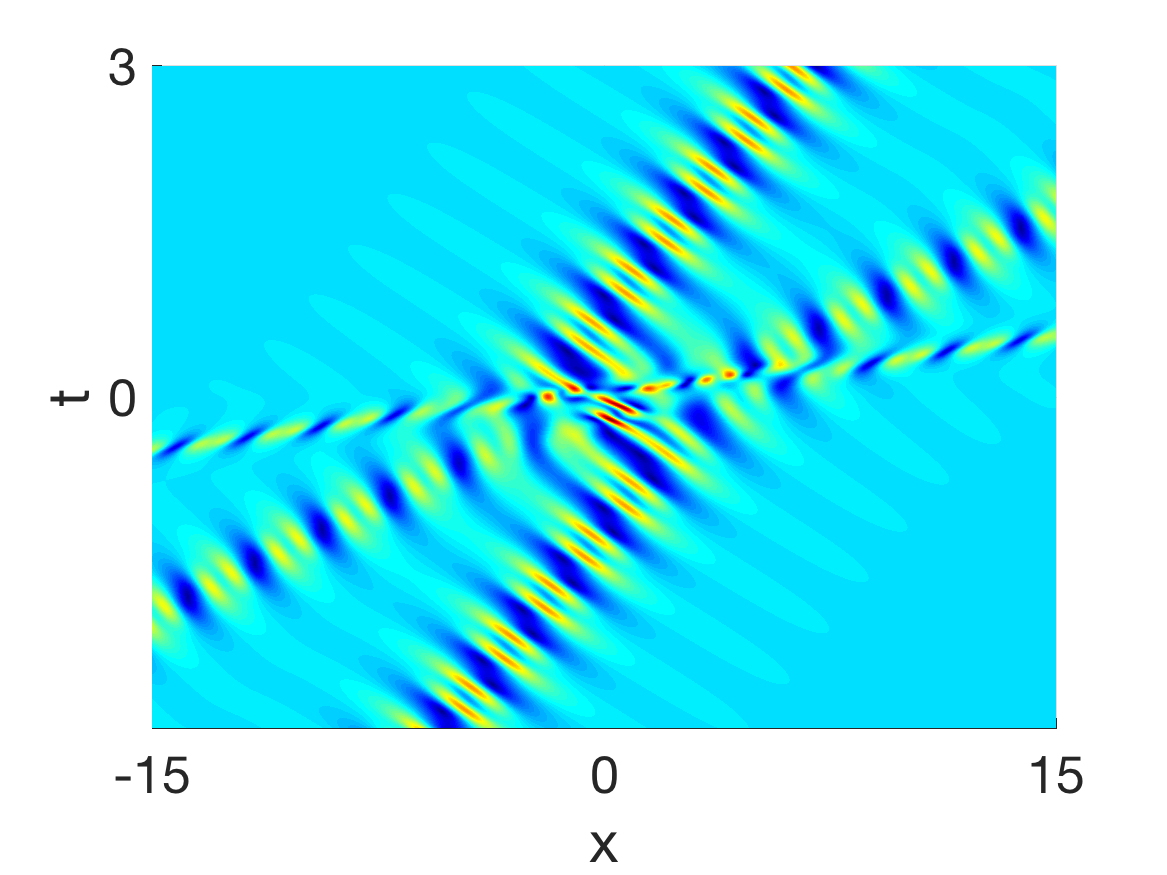}
      \label{3rd order breather-2-5}}
\subfigure[]{%
      \includegraphics[width=43mm]{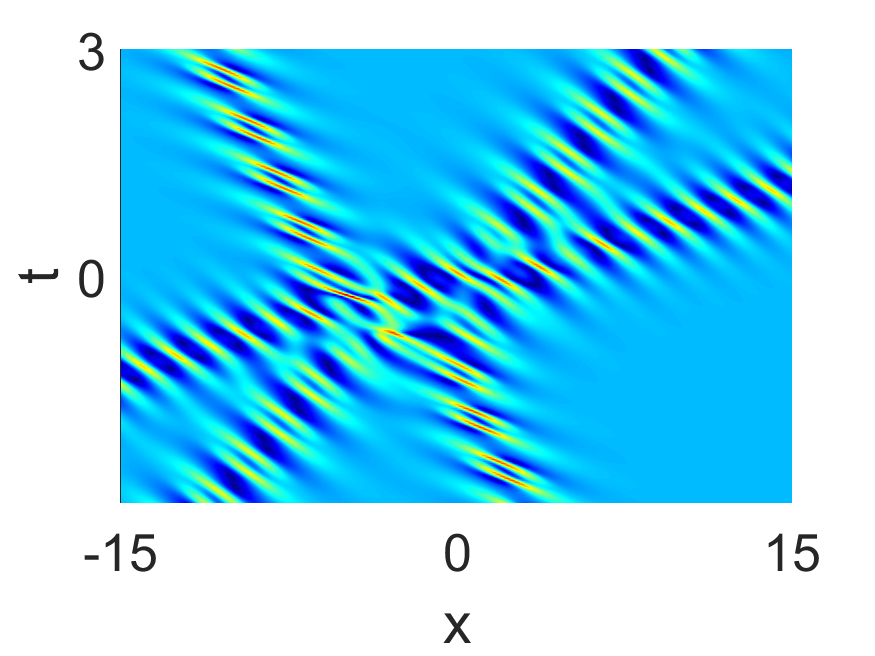}
      \label{3rd order breather-2-6}}
  \caption{(Color online) Third-order breather solutions with parameter values \(c=-1,\kappa=-1/2,\xi_{11,0}=\xi_{12,0}=0,\xi_{21,0}=\xi_{22,0}=0,\xi_{32,0}=0\) and (a) \(p_{11}= 1+1.7\mathbf{i}\), \(p_{21}= -0.65+2.5\mathbf{i}\), \(p_{31}= 0.95+1.65\mathbf{i},\) \( \xi_{31,0}= 2 \), (b) \(p_{11}= 0.95+1.65\mathbf{i}\), \(p_{21}= 0.8+2 \mathbf{i}\), \(p_{31}= 0.8+3.2\mathbf{i}, \) \( \xi_{31,0}= 2\), (c) \(p_{11}= 0.8+1.6\mathbf{i}\), \(p_{21}= -0.65+2\mathbf{i}\), \(p_{31}= 1.3+1.3\mathbf{i}, \) \( \xi_{31,0}= 4\),  where \(p_{12}\), \(p_{22}\) and \(p_{32}\) are given by \eqref{p12}. (d), (e) and (f) are the corresponding density plots of (a), (b) and (c), respectively. }
  \label{3rd order breather-1}
\end{figure*}

\vskip6pt

\enlargethispage{20pt}

\section*{Acknowledgement}
C.F. Wu was supported by the National Natural Science Foundation of China (Grant No. 11701382) and Guangdong Basic and Applied Basic Research Foundation, China (Grant No. 2021A1515010054). B.F. Feng was partially supported by National Science Foundation (NSF) under Grant No. DMS-1715991 and U.S. Department of Defense (DoD), Air Force for Scientific Research (AFOSR) under grant No. W911NF2010276.

\section*{Appendix}

In this appendix, we show that the first-order breather solutions presented in Theorem \ref{thm} can be expressed in terms of trigonometric functions and hyperbolic functions.
Let $N=1$, then Theorem \ref{thm} yields that
\begin{eqnarray}
%\begin{aligned}
\tau_0 &=& \left|
\begin{array}{cc}
    \sum\limits_{m=1}^{2} \sum\limits_{n=1}^{2} \frac{1}{p_{1m} + p_{1n}} e^{\xi_{1m} + \xi_{1n}} &\sum\limits_{m=1}^{2} \sum\limits_{n=1}^{2} \frac{1}{p_{1m} + p_{1n}^{*}} e^{\xi_{1m} + \xi_{1n}^{*}}  \\
    \sum\limits_{m=1}^{2} \sum\limits_{n=1}^{2} \frac{1}{p_{1m}^{*} + p_{1n}} e^{\xi_{1m}^{*} + \xi_{1n}} & \sum\limits_{m=1}^{2} \sum\limits_{n=1}^{2} \frac{1}{p_{1m}^{*} + p_{1n}^{*}} e^{\xi_{1m}^{*} + \xi_{1n}^{*}} \\
\end{array}
\right|\\
&=& |m_1|^2- m_2^2,
%\end{aligned}
\end{eqnarray}
where
\begin{eqnarray*}
  m_1 = \sum\limits_{m=1}^{2} \sum\limits_{n=1}^{2} \frac{1}{p_{1m} + p_{1n}} e^{\xi_{1m} + \xi_{1n}}, \quad   m_2 = \sum\limits_{m=1}^{2} \sum\limits_{n=1}^{2} \frac{1}{p_{1m} + p_{1n}^{*}} e^{\xi_{1m} + \xi_{1n}^{*}}.
\end{eqnarray*}
Denote by $ p_{11}= A+\mathbf{i} B,  p_{12}=   R+\mathbf{i} S, \xi_{11,0} = \alpha_1 + \mathbf{i} \beta_1, \xi_{12,0} = \alpha_2 + \mathbf{i} \beta_2$, where $p_{11}$ and $p_{12}$ satisfy \eqref{parameter constraint1}, then we have
\begin{eqnarray*}
  &&\xi_{11} = X+ \mathbf{i} Y =   A x +  (A^3-3 A B^2) t   +\alpha_1  + \mathbf{i} \left[ B x+ \left(3 A^2 B-B^3\right)t +\beta_1\right] ,
  \\
  && \xi_{12} = W+\mathbf{i} V = R x + (R^3-3 R S^2) t + \alpha_2+\mathbf{i} \left[S x +  \left(3 R^2 S-S^3\right) t +\beta_2 \right]  .
\end{eqnarray*}
After some tedious algebra, we can rewrite $\tau_0$ in the form
\begin{eqnarray}
  \tau_0 %&=& -\frac{S^2}{4 R^2 \left(R^2+S^2\right)} e^{4 W}   -\frac{B^2}{4 A^2 \left(A^2+B^2\right)} e^{4 X} + [c_1\cos (V-Y)+d_1 \sin  (V-Y)] e^{3 W+X}  \\
%    &&  + [c_2\cos (V-Y)+d_2 \sin  (V-Y)] e^{ W+3X} + e^{2 W+2X} [c_0 + c_3 \cos (2 (V-Y))+ d_3 \sin (2 (V-Y))]\\
    %&=& e^{2 W+2X}\left\{ -\frac{S^2}{4 R^2 \left(R^2+S^2\right)} e^{2 W-2X}   -\frac{B^2}{4 A^2 \left(A^2+B^2\right)} e^{-(2W- 2X)} + [c_1\cos (V-Y)+d_1 \sin  (V-Y)] e^{W-X} \right. \\
%    && \left. + [c_2\cos (V-Y)+d_2 \sin  (V-Y)] e^{-(W-X)} +  c_0 + c_3 \cos (2 V-2Y)+ d_3 \sin (2V-2Y) \right\} \\
    &=& e^{2 W+2X}\left\{ c_0  + M_1\cosh (2 W-2X - \theta_1)  \right. \nonumber\\
     &&+(c_1 + c_2) \cos (V-Y) \cosh (W-X) +  (c_1 - c_2) \cos (V-Y) \sinh (W-X)  \nonumber\\
    && + (d_1 + d_2) \sin (V-Y) \cosh (W-X) +  (d_1 - d_2) \sin (V-Y) \sinh (W-X) \nonumber\\
    && \left.+ (c_3+d_3)^{1/2} \cos (2 V-2Y - \theta_2)  \right\}  \label{tau0}
\end{eqnarray}
 where

\begin{tabular}{ll}
%\renewcommand\arraystretch{2}
%  \hline
  $ c_1  =  \dfrac{2 (R (A+R)+S (B+S))}{\left(R^2+S^2\right) K}-\dfrac{2 (A+R)}{R L }$, &
$  d_1  = \dfrac{2 (A S-B R)}{\left(R^2+S^2\right) K}-\dfrac{2 (S-B)}{RL}$, \\
$ c_2  = \dfrac{2 (A (A+R)+B (B+S))}{\left(A^2+B^2\right) K}-\dfrac{2 (A+R)}{A L }$,  &
 $  d_2  = \dfrac{2 (A S-B R)}{\left(A^2+B^2\right) K}-\dfrac{2 (S-B)}{A L}$, \\
   $c_3  = \dfrac{A R+B S}{2 \left(A^2+B^2\right) \left(R^2+S^2\right)}-\dfrac{2(A+R)^2-2(B-S)^2}{L^2}$, &
   \\
$    d_3 =  \dfrac{A S-B R}{2 \left(A^2+B^2\right) \left(R^2+S^2\right)}+\dfrac{4 (A+R) (B-S)}{L^2}$, &
\\
  $ c_0 =\dfrac{4}{K} -\dfrac{2}{L}-\dfrac{1}{2A R}$, $ \qquad K = (A+R)^2+(B+S)^2 $, & $ L=(A+R)^2+(B-S)^2$,
  \\
$  M_1 = 2 \sqrt{\sigma_1 \sigma_2}$, $ \qquad \theta_1 = \ln (\sigma_2/\sigma_1)/2 $, $  \qquad \theta_2 = \arctan (d_3/c_3) $, &
  \\
$ \sigma_1 =  \dfrac{S^2}{4 R^2 \left(R^2+S^2\right)}$, $ \qquad \sigma_2 =  \dfrac{B^2}{4 A^2 \left(A^2+B^2\right)}  $. &
 % \hline
\end{tabular}

Similarly, we have
\begin{eqnarray}
  \tau_1
  &=& e^{2 W+2X}\left\{ e_0  + M_2\cosh (2 W-2X - \theta_3)  \right.\nonumber\\
  &&+(e_1 + e_2) \cos (V-Y) \cosh (W-X) +  (e_1 - e_2) \cos (V-Y) \sinh (W-X) \nonumber\\
  && + (f_1 + f_2) \sin (V-Y) \cosh (W-X) +  (f_1 - f_2) \sin (V-Y) \sinh (W-X)\nonumber\\
  && + (e_3+f_3)^{1/2} \cos (2 V-2Y - \theta_4) \nonumber\\
  &&+i\left[\hat{e}_0  + \hat{M}_2\cosh (2 W-2X - \hat{\theta}_3)  \right.\nonumber\\
  &&+(\hat{e}_1 + \hat{e}_2) \cos (V-Y) \cosh (W-X) +  (\hat{e}_1 - \hat{e}_2) \cos (V-Y) \sinh (W-X) \nonumber\\
  && + (\hat{f}_1 + \hat{f}_2) \sin (V-Y) \cosh (W-X) +  (\hat{f}_1 - \hat{f}_2) \sin (V-Y) \sinh (W-X)\nonumber\\
  && \left.\left. + (\hat{e}_3+\hat{f}_3)^{1/2} \cos (2 V-2Y - \hat{\theta}_4)\right]\right\}, \label{tau1}
\end{eqnarray}
where
\begin{eqnarray*}
  e_0 &=& \Re(K_{22}+K_{23}+K_{32}+K_{33}-L_{14}-L_{22}-L_{33}-L_{41}),\\
  e_1 &=& \Re(K_{24}+K_{34}+K_{42}+K_{43}-L_{24}-L_{43}-L_{42}-L_{34}),\\
  f_1 &=& \Im(-K_{24}-K_{34}+K_{42}+K_{43}+L_{24}+L_{43}-L_{42}-L_{34}),\\
  e_2 &=& \Re(K_{12}+K_{13}+K_{21}+K_{31}-L_{13}-L_{21}-L_{31}-L_{12}),\\
  f_2 &=& \Im(-K_{12}-K_{13}+K_{21}+K_{31}+L_{13}+L_{21}-L_{31}-L_{12}),\\
  e_3 &=& \Re(K_{14}+K_{41}-L_{23}-L_{32}),\\
  f_3 &=& \Im(-K_{14}+K_{41}-L_{23}+L_{32}),\\
  \hat{e}_0 &=& \Im(K_{22}+K_{23}+K_{32}+K_{33}-L_{14}-L_{22}-L_{33}-L_{41}),\\
  \hat{e}_1 &=& \Im(K_{24}+K_{34}+K_{42}+K_{43}-L_{24}-L_{43}-L_{42}-L_{34}),\\
  \hat{f}_1 &=& \Re(K_{24}+K_{34}-K_{42}-K_{43}-L_{24}-L_{43}+L_{42}+L_{34}),\\
  \hat{e}_2 &=& \Im(K_{12}+K_{13}+K_{21}+K_{31}-L_{13}-L_{21}-L_{31}-L_{12}),\\
  \hat{f}_2 &=& \Re(K_{12}+K_{13}-K_{21}-K_{31}-L_{13}-L_{21}+L_{31}+L_{12}),\\
  \hat{e}_3 &=& \Im(K_{14}+K_{41}-L_{23}-L_{32}),\\
  \hat{f}_3 &=& \Re(K_{14}-K_{41}+L_{23}-L_{32}),\\
  M_2 &=& 2\sqrt{\Re((K_{11}-L_{11})(K_{44}-L_{44}))}, \quad  \hat{M}_2 = 2\sqrt{\Im((K_{11}-L_{11})(K_{44}-L_{44}))},\\
  \theta_3 &=& \ln \Re((K_{11}-L_{11})/(K_{44}-L_{44}))/2, \quad \hat{\theta}_3 = \ln \Im((K_{11}-L_{11})/(K_{44}-L_{44}))/2,\\
  \theta_4 &=& \arctan (f_3/e_3), \quad \hat{\theta}_4 = \arctan (\hat{f}_3/\hat{e}_3),
\end{eqnarray*}
and \(K_{kl}, L_{kl}, (k,l=1,2,3,4)\) are given by
\begin{equation}
  K_{kl} = n^{(k)}_{11} \times n^{(l)}_{22}, \quad L_{kl} = n^{(k)}_{12} \times n^{(l)}_{21},
\end{equation}
with $(i,j=1,2)$
\begin{eqnarray}
  n^{(1)}_{ij} = \frac{(\mathbf{i} \kappa -p_{i,1})}{(p_{j,1}+a) (p_{i,1}+p_{j,1})},\quad n^{(2)}_{ij} = \frac{(\mathbf{i} \kappa-p_{i,1})}{(p_{j,2}+a) (p_{i,1}+p_{j,2})},\\
  n^{(3)}_{ij} = \frac{(\mathbf{i} \kappa-p_{i,2})}{(p_{j,1}+a) (p_{i,2}+p_{j,1})},\quad n^{(4)}_{ij} = \frac{(\mathbf{i} \kappa-p_{i,2})}{(p_{j,2}+a) (p_{i,2}+p_{j,2})}.
\end{eqnarray}
As a consequence, the first-order breather solutions of the Sasa-Satsuma  equation \eqref{SS equation} can be rewritten as
$$
u=\frac{\tau_{1} (x-6 c t,t)}{\tau_{0} (x-6 c t,t)} e^{\mathrm{i}\left(\kappa(x-6 c t)-\kappa^{3} t\right)},
$$
where $\tau_{0}$ and $\tau_{1}$ are given by \eqref{tau0} and \eqref{tau1} respectively.

%\newpage
%
%{\color{red} 2021.7.15}

%\end{document}

%\begin{thebibliography}{9}
%
%\bibitem{1} Allwood JM, Cullen JM. 2011 \textit{Sustainable materials:  with both eyes open}.
%Cambridge, UK: UIT Cambridge. See \href{http://www.withbotheyesopen.com}{http://www.withbotheyesopen.com}.
%
%\bibitem{2}  MacKay DJC. 2008  \textit{Sustainable energy:  without the hot air}.
% Cambridge, UK: UIT Cambridge. See \href{http://www.withouthotair.com}{http://www.withouthotair.com}.
%
%\bibitem{3} Gallman PG. 2011  \textit{Green alternatives and national energy strategy: the facts
% behind the headlines}.  Baltimore,\ MD: Johns Hopkins University Press.
%
%\bibitem{4} MacKay DJC. 2013.  Solar energy in the context of energy use, energy transportation, and
% energy storage. \textit{Proc. R. Soc. A} \textbf{371}.
%
%\end{thebibliography}

\end{document}